\documentclass[12 pt]{article}
\usepackage{subfiles}
\usepackage{bbm}
\usepackage{natbib}
\usepackage{url}
\usepackage{amssymb}
\usepackage{amsmath}
\usepackage{amsthm}
\usepackage[onehalfspacing]{setspace}
\usepackage[colorlinks]{hyperref}
\usepackage{graphicx}
\usepackage{subcaption}
\usepackage{enumitem}
\usepackage{verbatim}
\usepackage{sectsty}
\usepackage[svgnames]{xcolor}
\usepackage{subfiles}
\usepackage[capitalize,nameinlink]{cleveref}
\usepackage[colorinlistoftodos,textsize=tiny]{todonotes}
\usepackage{xcolor}
\usepackage{pgf}
\usepackage{tikz}
\usetikzlibrary{patterns, decorations.pathreplacing, arrows.meta}
\usepackage{xargs}
\usepackage[font={footnotesize,it},justification=justified]{caption}
\usepackage[normalem]{ulem}
\usepackage{booktabs}
\hypersetup{
    pdftitle={Voluntary Disclosure and Personalized Pricing},
    pdfauthor={Nageeb Ali, Greg Lewis, Susana Vasserman},
    citecolor=DarkBlue,
    bookmarksnumbered=true,
    urlcolor=Indigo,linkcolor=DarkBlue
}
 \theoremstyle{remark}
 \newtheorem{example}{Example}

\theoremstyle{plain}

\newtheorem{lemma}{Lemma}

\newtheorem{observation}{Observation}

\newtheorem{proposition}{Proposition}

\renewcommand{\qedsymbol}{\textsquare}
\renewcommand{\epsilon}{\varepsilon}
\usepackage[margin=0.9in]{geometry}
\newcommandx{\nageeb}[2][1=]{\todo[linecolor=blue,backgroundcolor=blue!25,bordercolor=blue,#1]{#2}}
\newcommandx{\greg}[2][1=]{\todo[linecolor=red,backgroundcolor=red!25,bordercolor=red,#1]{#2}}  
\newcommandx{\shosh}[2][1=]{\todo[linecolor=green,backgroundcolor=green!25,bordercolor=green,#1]{#2}}
\sectionfont{\color{DarkRed}}
\subsectionfont{\color{DarkRed}}
\subsubsectionfont{\color{DarkRed}}
\crefname{assumption}{assumption}{assumptions}

\newcommand{\messages}{\mathcal M}
\newcommand{\partition}{\mathcal P}
\newcommand{\lowv}{\underline{v}}
\newcommand{\highv}{\overline{v}}

\newcommand{\optprice}{p^*}

\begin{document}

\begin{titlepage}

\title{Voluntary Disclosure and Personalized Pricing\thanks{We thank Nima Haghpanah, Navin Kartik, Vijay Krishna, Doron Ravid, Nikhil Vellodi, Jidong Zhou, and various conferences and seminars audiences. We are grateful to Microsoft Research for hosting Ali and Vasserman. We thank Xiao Lin for excellent research assistance.}}
\author{
S. Nageeb Ali\thanks{Pennsylvania State University. Email: \href{mailto:nageeb@psu.edu}{nageeb@psu.edu}.}
\and
Greg Lewis\thanks{Microsoft Research. Email: \href{glewis@microsoft.com}{glewis@microsoft.com}}
\and 
Shoshana Vasserman\thanks{Stanford Graduate School of Business.  Email: \href{mailto:svass@stanford.edu}{svass@stanford.edu}.}
}
\date{August 15, 2020}
\maketitle

\begin{abstract}
Central to privacy concerns is that firms may use consumer data to price discriminate. A common policy response is that consumers should be given control over which firms access their data and how. Since firms learn about a consumer's preferences based on the data seen and the consumer's disclosure choices, the equilibrium implications of consumer control are unclear. We study whether such measures improve consumer welfare in monopolistic and competitive markets. We find that consumer control can improve consumer welfare relative to both perfect price discrimination and no personalized pricing. First, consumers can use disclosure to amplify competitive forces. Second, consumers can disclose information to induce even a monopolist to lower prices. Whether consumer control improves welfare depends on the disclosure technology and market competitiveness. Simple disclosure technologies suffice in competitive markets. When facing a monopolist, a consumer needs partial disclosure possibilities to obtain any welfare gains.  

\medskip


\end{abstract}
\thispagestyle{empty} 

\end{titlepage}
\clearpage
	\setcounter{tocdepth}{2}
	\tableofcontents
	\thispagestyle{empty}
\clearpage

\setcounter{page}{1}

\maketitle

\section{Introduction}\label{Section-Introduction}

Consumer data is valuable information.  Sellers want to know what consumers are interested in, which products they are aware of, and how much they might be willing to pay for them.
Answers to these questions can be collected to an unprecedented degree through online activity. The resulting information can then be used to personalize both prices and product offerings by online sellers.
But many consumer advocates have argued that consumer data is inherently \emph{private}, and that  consumers should be able to choose whether or not to share it with companies. This tension has led regulators to pose a question that made little sense in the physical economy: to what degree are consumers in the online economy harmed by the release of their private information to sellers, and should consumers be able to control which information they share?

These questions are at the forefront of an ongoing international debate that has precipitated action in both the public and private sectors.
Regulators concerned that personalized pricing will harm consumers have focused on the importance of consumer consent, passing wide-reaching legislation on data storage and tracking. A prominent example, the General Data Protection Regulation (GDPR) passed by the European Union, requires firms to obtain consent from consumers before collecting and processing their personal data.\footnote{Starting on January 1, 2020, California has enforced the California Consumer Privacy Act (CCPA), which has similar provisions to the GDPR.} In the United States, the Federal Trade Commission recommends that ``best practices include...giving consumers greater control over the collection and use of their personal data...'' \citep{FTC2012}. Meanwhile, private sector firms have responded to consumer demand for privacy by designing commercial products and brands that are specifically developed to limit tracking.\footnote{For example, Apple recently added a feature to its Safari browser that limits the ways in which its user's activities are tracked by third parties \citep{Hern2018}.} 

Against this backdrop, we study the market implications of consumer consent and control. 
We investigate what happens when consumers fully control their data---not only whether they are tracked, but what specific information is disclosed to firms.
We assume that the data is encoded in a verifiable form that consumers can partially or fully disclose to firms.
In equilibrium, firms draw inferences about a consumer's preferences based on both what she shares and what she chooses not to share.
Our motivating question is: \emph{when consumers fully control their information, are they hurt or helped by personalized pricing?}

We pose this question in an environment in which products cannot be personalized, and so there is no match value from data.
A classical intuition might suggest that consumers would not benefit from being permitted to voluntarily disclose information. Because the market's \emph{equilibrium inferences} are based both on information that is disclosed and what is not being disclosed, giving consumers the ability to separate themselves may be self-defeating, as seen in the unraveling equilibria of \cite{grossman:81-JLE} and \cite{milgrom1981good}. Contrary to that intuition, we find that the combination of personalized pricing and consumer control is actually beneficial to consumers in both monopolistic and competitive markets.
We construct equilibria of the consumers' disclosure game in which sharing data weakly increases consumer surplus for \emph{every} consumer type, relative to the benchmark of no personalized pricing.

Two key ideas drive these results.
First, voluntary disclosure and personalized pricing together amplify competition between firms.
Nearly indifferent consumers benefit from the ability to credibly communicate their flexibility, intensifying competition for their business, while consumers with a strong preference for the product of one particular firm can hide this preference. Although firms interpret this non-disclosure as a signal of strong preferences, the resulting prices are still lower than without disclosure.
Second, even in the absence of competition, consumers can benefit from sending coarse signals that pool their valuations.
These pools are constructed in such a way that a monopolist finds it optimal to sell to every type within each pool. Therefore every consumer within a pool pays the price of the consumer type that has the lowest valuation in that pool. Disclosures lead to price discounts that benefit every consumer type. 
The take-away is that offering consumers control---and possibly building tools that coordinate the sharing of data for consumer benefit---may make personalized pricing attractive \emph{even in the absence of better matching.} 

\paragraph{A Preview:}Although we are interested in both monopolistic and competitive markets, we start with the problem of a monopolist choosing what price to charge a consumer whose valuation he does not know. We augment that classical problem with a ``verifiable'' disclosure game, as in \cite{grossman:81-JLE} and \cite{milgrom1981good}: before the monopolist sets his price, a consumer chooses what ``evidence'' or hard information to disclose about her valuation. We study both those disclosure environments in which evidence is \textbf{simple}, where a consumer can either speak ``the whole truth'' (reveal her type) or say nothing at all, as well as those in which evidence is \textbf{rich}, where a consumer can disclose facts about her type without having to reveal it completely.\footnote{We borrow this terminology from \cite{hagenbach2017}.} 
In our game, the consumer first observes her type and then chooses a message to disclose to the firm from the set of messages available to her. The firm then quotes a price, and the consumer chooses whether to buy the product at that price. The firm and the consumer cannot commit to their strategic choices. 

Our first conclusion for the monopolistic environment (\Cref{Proposition-Certification}) is that simple evidence never benefits the consumer and potentially hurts her: there is no equilibrium in which \emph{any type} of the consumer is better off relative to a setting without personalized pricing. Moreover, there are equilibria in which all consumer types are worse off, such as an unraveling equilibrium in which the monopolist extracts all surplus. 

Our second conclusion is that once the evidence structure is rich---where consumers can partially disclose information without revealing all of it---all consumer types can benefit from disclosing information. 
\Cref{Proposition-MonopolistConsumerSegmentation} constructs an equilibrium that improves the consumer surplus for almost all consumer types without reducing the surplus of any consumer type. In this equilibrium, all types are partitioned into segments on the basis of their willingness to pay, and trading is fully efficient. Because the consumer cannot commit to her disclosure strategy, every consumer type must find that her equilibrium message induces a weakly lower price than that induced by any other feasible message; our segmentation guarantees this property. Moreover, our segmentation ensures that for each segment, the monopolist's optimal price is the lowest willingness to pay in that segment. This greedy algorithm identifies a Pareto-improving equilibrium segmentation for every distribution of consumer types and identifies the equilibrium segmentation that maximizes ex ante consumer surplus for a class of distributions. 

We show that stronger conclusions emerge in competitive markets. We consider general competitive markets with differentiated products, in which the consumer demands a single unit. The differentiation may be purely horizontal or have both horizontal and vertical components. The firms are uncertain of the consumer's preferences and the consumer can disclose information about her preferences to the firms, who then simultaneously make price offers to her. As before, we compare the outcomes when consumers can disclose information --- either via simple or rich evidence --- against a benchmark model in which there is no personalized pricing. Here, voluntary disclosure and personalized pricing is particularly beneficial to consumer surplus because of a new economic force: \emph{information can be selectively disclosed to amplify competition}.

More specifically, if the distribution of consumer preferences is symmetric and log-concave, then an equilibrium in the game with simple evidence (where the consumer's disclosure strategies are all-or-nothing) improves consumer welfare for every type relative to the no-personalized-pricing benchmark. Rich evidence allows one to do even better by using a greedy segmentation strategy similar to that used in the monopolist problem. But in a competitive market, rich evidence is unnecessary for these Pareto gains, and simple evidence suffices.

\paragraph{Implications:} 
The goal of our exercise is to use a simple model to inform the ongoing regulatory debate on consumer consent. From this stylized model, we draw two broad lessons for policy. First, voluntary disclosure can facilitate price concessions in both monopolistic and competitive markets. Thus, there is something missing in the view that tracking involves a tradeoff between the benefits of personalized products and the costs of personalized prices. Even without the benefits of product personalization, a consumer can benefit from personalized pricing when she has control. Disclosure generates discounts and amplifies competition.

The second lesson is that whether a track / do-not-track regime (as evoked by the GDPR) suffices to give consumers \emph{useful} control over their data depends on the competitiveness of the marketplace. In a monopolistic environment, richer forms of data sharing are necessary for consumers to benefit: consumer control must involve a choice not only of \textit{whether} to share information but also of \textit{how much} information to share. By contrast, in competitive markets, the choice of whether to share information is sufficient so long as consumers can share information with some firms and conceal it from others. Because both monopolistic and competitive environments feature multiple equilibria, some degree of coordination is needed to orchestrate which information to group together. 

While online communications between consumers and sellers are not yet as sophisticated as is envisioned in our rich-evidence setup, an important element of the digital economy is its increasing ability to verify information \citep{goldfarbtucker}. These advances suggest that it may be technologically feasible for consumers to use intermediaries or platforms to verifiably disclose that their preferences or characteristics (e.g., income, age, address etc.) lie within a certain range without having to forfeit all of their information to online sellers.  Regulatory agencies may seek to constrain the ways in which such intermediaries confer information between buyers and sellers, or to provide information verification services themselves. Our work provides suggestions for what an effective intermediary of this form might look like.

\paragraph{Relationship to Literature:}
Our work belongs to the growing literatures on privacy, information, and their implications for markets; see \cite*{acquisti2016economics} and \cite{bergemann2019survey} for recent surveys. We view our paper as making two contributions. First, it formulates and investigates the economic implications of giving consumers control over their data through the lens of voluntary disclosure. 
Second, our analysis shows that whether consumers benefit from controlling their information depends both on the technology of disclosure and the level of market competition. 

We combine classical models of monopolistic and competitive pricing with the now classical study of verifiable disclosure. Unlike the first analyses of verifiable disclosure \citep{grossman:81-JLE,milgrom1981good}, unraveling is not the unique equilibrium outcome of the markets that we study. An observation central to our results is that a firm's optimal price need not be monotone in its beliefs about the consumer's willingness to pay. This observation permits us to pool high and low types without giving the low type an incentive to separate.\footnote{Prior analyses have highlighted other reasons for why markets may not unravel, in particular (i) uncertainty about whether the sender has evidence \citep{dye1985disclosure,shin1994burden}, (ii) disclosure costs \citep{jovanovic1982truthful,verrecchia1983discretionary}, or (iii) the possibility for receivers to be naive \citep{hagenbach2017}.} 

For the case of a monopolistic seller, two papers in this verifiable disclosure literature touch upon closely related issues. \cite{sher2015price} study monopolistic price discrimination in which the seller commits to a schedule of evidence-contingent prices. By contrast, we assume that sellers cannot commit and instead set prices that best respond to the evidence that has been presented. 
In independent and prior work, \cite{pram2020} studies when a consumer can disclose rich evidence to obtain Pareto gains from a monopolistic seller.\footnote{We became aware of his work subsequent to our first draft.} His paper and ours share some overlap in studying the rich-evidence case of a monopolist facing a consumer that has private values, but we go in different directions from that benchmark. Pram's interest is in adverse-selection settings where the payoffs for the monopolist and the buyer depend on the buyer's type, and he elegantly characterizes necessary and sufficient conditions for the existence of a Pareto improving equilibrium.\footnote{Other papers that study the role of certification in adverse selection are \cite{stahl2017certification} and \cite*{glode2018voluntary}.} By contrast, our motivation is to understand the interaction of disclosure technologies and market competition, which is why we study both simple and rich evidence in both monopolistic and competitive markets. Our message of how simple evidence amplifies competitive forces but richer forms of evidence are needed with a monopolist does not feature in his work. We view our papers to be highly complementary shedding light on different aspects of consumer control.

This verifiable-disclosure approach to consumer control complements two other ways for consumers to share information in markets. One way is through information design, where an intermediary that knows the consumer's type commits to a segmentation strategy. With a monopolistic seller, that intermediary can achieve payoffs characterized by \cite*{bergemann2015limits}. 
The other way is for the consumer to use cheap talk. \cite{hidir2018personalization} show that if the product can be customized to the consumer's tastes, cheap talk can improve matching without resulting in the monopolist capturing the entire surplus.
Verifiable disclosure offers a potentially useful middle ground between information design and cheap talk. It is relevant when a consumer can use an intermediary to verify information about her type in her communication to firms without forfeiting control over the disclosure of that information. As the evolving digital economy balances an increasing ability to verify information cryptographically with public pressure for individual privacy, we believe it to be useful to understand the possible implications when consumers can communicate verifiable information to firms. 
 
The work discussed above involves a single seller but much of our interest is in how voluntary disclosure amplifies competition in markets with differentiated goods. Related to this idea is the innovative work of \cite{thissevives1988}. Using a model of Bertrand duopoly, they study whether firms choose to personalize prices. In their model, the consumer's type is commonly known and they show that the unique equilibrium involves personalized pricing even though the firms' joint profits would be higher if they could commit to uniform pricing. The key distinction with our work is that all of the action in their model --- and in the subsequent literature --- is on the side of the firms: taking consumers as passive, these papers study whether firms personalize prices when they know (or can learn) consumer types.\footnote{See \cite{armstrong_2006} for a survey. We thank Jidong Zhou for drawing our attention to this work. }  
 By contrast, in our model, the consumer actively chooses whether to disclose information and it is her voluntary disclosure that facilitates personalized pricing. Moreover, the ability to pool with other types is necessary for every consumer type to benefit from personalized pricing; otherwise, extreme types are worse off from personalized pricing. Thus, the welfare gains that we study would not emerge in a model in which consumer types are commonly known. 

 Several papers share our interest in the role of information in competitive markets with differentiated products. \cite{elliott2019market} show how an information-designer can segment the market so that consumers are allocated efficiently while nevertheless guaranteeing that consumers obtain no surplus. \cite{armstrongzhou2019} study firm-optimal and consumer-optimal information structures for a consumer that does not know her tastes, and show that the consumer-optimal signal may involve learning a little so as to amplify price competition. Thinking about network information, \cite{fainmesser2015pricing,fainmesser2019pricing} study monopolistic and competitive price discrimination based on how consumers influence each other. 
 
Less closely related, a prior literature studies how a consumer may distort her behavior in dynamic settings if firms draw inferences about her tastes from her past choices. This literature has investigated when firms prefer to commit to not personalize prices, when consumers would like to remain anonymous, and how consumers may distort their actions \citep*{taylor2004consumer,villasboas:04,acquisti2005conditioning,calzolari2006optimality,conitzer2012hide}. \cite{bonatti/cisternas} study the welfare properties of aggregating consumers' past purchasing histories into scores.
Our analysis complements this agenda by studying how a consumer fares from directly controlling the flow of information rather than distorting her behavior to influence the market's perception of her tastes. 
 
 In thinking about consumer control, we abstract from issues raised in the study of data markets. Several recent papers \citep*{choi2019privacy,acemoglu2019too,bergemann2019economics} model how the data of some consumers may be predictive of others but each consumer may not internalize this externality, inducing excessive data-sharing. \cite{liang2020data} study how data linkages across individuals affect their incentives for exerting effort. \cite{jones2019nonrivalry} show that, because data is non-rival, there are social gains from multiple firms using the same data simultaneously, and therefore it is better to let consumers own and trade data. \cite*{fainmesser2020digital} study issues of data collection and data protection, and tradeoffs thereof.
 
\paragraph{Outline:} Using  examples, we illustrate our main results for both monopolistic and competitive markets in \Cref{Section-Example}. \cref{Section-MonopolistModel} offers the general analysis for the monopolist setting.  In \cref{Section-Competition}, we analyze competitive markets with product differentiation. Therein, we study both Bertrand duopoly with horizontal differentiation as well as a general model with $n\geq 2$ firms and general product differentiation. \cref{Section-Conclusion} concludes.

\section{Two Examples}\label{Section-Example}
 
\subsection{A Monopolist}\label{Section-ExampleMonopolist}
A monopolist (``he'') sells a good to a single consumer (``she''), who demands a single unit. The consumer's value for that good is $v$, which is drawn uniformly from $[0,1]$. If the consumer purchases the good from the monopolist at price $p$, her payoff is $v-p$ and the monopolist's payoff is $p$; otherwise, each party receives a payoff of $0$. The consumer knows her valuation for the good and the monopolist does not. In this setting, and without any disclosure, the monopolist optimally posts a uniform price of $\frac{1}{2}$, which induces an ex interim consumer surplus of $\max\{v-\frac{1}{2},0\}$, and a producer surplus of $\frac{1}{4}$. 

We augment this standard pricing problem with voluntary disclosure on the part of the consumer. After observing her value, the consumer chooses a message $m$ from the set of feasible messages for her. 
The set of \emph{all} feasible messages is $\messages\equiv \left\{[a,b]:0\leq a \leq b \leq 1\right\}$, and we interpret a message $[a,b]$ as \emph{``My type is in the set $[a,b]$.''} When a consumer's type is $v$, the set of messages that she can send is $M(v)\subseteq \messages$. The evidence structure is represented by the correspondence $M:[0,1]\rightrightarrows \messages$. 

Here is the timeline for the game: first, the consumer observes her type $v$ and chooses a message $m$ from $M(v)$. The monopolist then observes the message and chooses a price $p\geq 0$. The consumer then chooses whether to purchase the good. Each party behaves sequentially rationally: we study Perfect Bayesian Equilibria (henceforth PBE) of this game. Our interest is in the implications of this model for simple and rich evidence structures, described below.

\paragraph{Simple evidence:}
An evidence structure is \textbf{simple} if for every $v$, $M(v)=\{\{v\},[0,1]\}$. In other words, each type $v$ can either fully reveal her type using the message $m=\{v\}$ (which is unavailable to every other type), or not disclose anything at all, using the message $m=[0,1]$ (which is available to every type). Such an evidence structure offers a stylized model for the dichotomy between ``track'' and ``do-not-track'': a consumer who opts into tracking will have all of her digital footprint observed by the buyer, whereas do-not-track obscures it entirely.

In this game, there exists an equilibrium in which every type $v$ fully reveals itself using the message $m=\{v\}$, and the monopolist extracts all surplus on the equilibrium path. Off-path, if the consumer sends the non-disclosure message, $m=[0,1]$, the monopolist believes that $v=1$ with probability $1$, and charges a price of $1$. In this equilibrium, all consumers are hurt by the possibility of voluntary disclosure and personalized pricing but the monopolist benefits from it.

But this is not the only equilibrium: there is also one in which every type sends the {non-disclosure} message $m=[0,1]$, and the monopolist charges a price of $\frac{1}{2}$. No consumer type wishes to deviate because revealing her true type results in a payoff of $0$. Here, both consumer and producer surplus are exactly as in the world without personalized pricing. In fact, there are an uncountable number of equilibria. But 
\emph{none} of them improve upon the benchmark of no-personalized-pricing from the perspective of \emph{any} consumer type.
\begin{observation}\label{Observation-Simple}
	With simple evidence, across all equilibria, the consumer's interim payoff is no more than her payoff without personalized pricing, namely $\max \{v-1/2,0\}$. 
\end{observation}  

\begin{figure}[t]
\centering
{\scriptsize
	\begin{tikzpicture}[scale=0.45]
		\draw[fill=SkyBlue,fill opacity=1, ultra thick] (0,7) rectangle (1,10); 
		\draw[pattern=crosshatch dots, pattern color=orange, ultra thick] (0,3) rectangle (1,7); 
		\draw[ultra thick] (0,0) rectangle (1,3); 
		\node at (1.6,0) {$0$};
		\node at (1.6,10) {$1$};
		\draw (1, 8.7)-- (1.3, 8.7);
		\node at (1.6, 8.7) {$v$};
		\draw (1, 5)--(1.3,5);
		\node  at (1.9, 5) {$p_{ND}$};
		\draw [->, thick] (1.6,8.4) to[bend left] (1.6,5.3);
		\node at (3.5, 7) {profitable};
		\node at (3.5, 6.3) {deviation};
		\node at (-2.3,8.7) {Full Disclosure};
		\node at (-2.3,5) {Non Disclosure};
		\node at (0.5,-1) {$(a)$};
	\end{tikzpicture}\hspace{0.5in}	
	\begin{tikzpicture}[scale=0.45]
		\draw[pattern=crosshatch dots, pattern color=orange, ultra thick] (0,6.5) rectangle (1,10); 
		\draw[ultra thick] (0,0) rectangle (1,6.5); 
		\node at (1.6,0) {$0$};
		\node at (1.6,10) {$1$};
		\draw (1, 6.5)--(1.3,6.5);
		\node  at (1.9, 6.5) {$p_{ND}$};
		\draw[thick, red] (0, 5)--(1.2,5);
		\node [red] at (1.57, 5) {$\frac{1}{2}$};
		\node at (0.5,-1) {$(b)$};
	\end{tikzpicture}\hspace{0.8in}
\begin{tikzpicture}[scale=0.45]
		\draw[pattern=crosshatch dots, pattern color=orange, ultra thick] (0,4) rectangle (1,10); 
		\draw[ultra thick] (0,2) rectangle (1,4); 
		\draw[pattern=crosshatch dots, pattern color=orange, ultra thick] (0,0) rectangle (1,2);
		\node at (1.6,0) {$0$};
		\node at (1.6,10) {$1$};
		\draw (1, 6.5)--(1.3,6.5);
		\draw[thick, red] (0, 5)--(1.2,5);
		\node at (2.5, 5) {$p_{ND}=\frac{1}{2}$};
		\node at (0.5,-1) {$(c)$};	
		\end{tikzpicture}
	}
\vspace{-0.1in}
\caption{$(a)$ shows that any disclosing type that is strictly higher than $p_{ND}$ has a profitable deviation $\Rightarrow$ the set of non-disclosing types includes $(p_{ND},1]$. $(b)$ and $(c)$ show different equilibria where the shaded region is the set of non-disclosing types. Across equilibria, $p_{ND}\geq 1/2$.}\label{Figure-SimpleEvidence}
\end{figure}
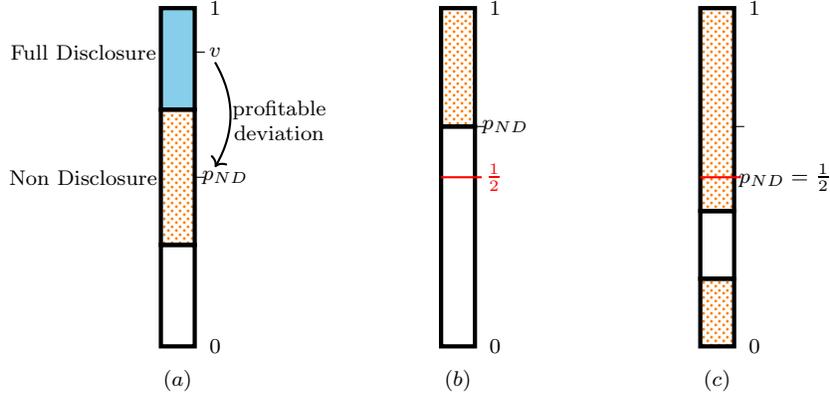

We illustrate the argument in \autoref{Figure-SimpleEvidence}. In an equilibrium where a positive mass send the non-disclosure message, suppose that the monopolist charges $p_{ND}$ when he receives this message. Any type $v$ that is strictly higher than $p_{ND}$ must send the non-disclosure message because her other option---revealing herself---induces a price that extracts all of her surplus (this property is shown in \autoref{Figure-SimpleEvidence}(a)). Hence, the set of types that send the non-disclosure message includes $(p_{ND},1]$. There are various configurations of disclosure and non-disclosure segments that are compatible with this requirement, as shown in (b) and (c), but across them all, the monopolist's optimal non-disclosure price $p_{ND}$ never dips below $\frac{1}{2}$, the price charged without personalized pricing. 

\paragraph{Rich evidence:} \Cref{Observation-Simple} illustrates that simple evidence structures and personalized pricing do not benefit the consumer. Now we see how the consumer can do better if she can use a rich evidence structure.
An evidence structure is \textbf{rich} if for every $v$, $M(v)=\{m\in\messages:v\in m\}$; in other words, a type $v$ can send any interval that contains $v$. With a rich evidence structure, all the equilibrium outcomes described above can still be supported using this richer language. But new possibilities emerge, some of which dominate the payoffs from no-personalized-pricing. 

We describe an equilibrium that strictly improves consumer surplus for a positive measure of consumer types without making any type worse off. Inspired by Zeno's Paradox,\footnote{Zeno's Paradox is summarized by Aristotle as \emph{``...that which is in locomotion must arrive at the half-way stage before it arrives at the goal....''} See \href{https://plato.stanford.edu/entries/paradox-zeno/}{https://plato.stanford.edu/entries/paradox-zeno/}.} consider the countable grid $\left\{1,\frac{1}{2},\frac{1}{4},\ldots\right\}\cup\left\{0\right\}$. We denote the $(k+1)^{th}$ element of this ordered list, namely $2^{-k}$, by $a_{k}$, and the set $m_k\equiv [a_{k+1},a_k]$. We use this partition to construct an equilibrium segmentation that improves consumer surplus. 
\begin{observation}\label{Observation-Rich}
With rich evidence, there exists an equilibrium that generates Zeno's Partition: a consumer's reporting strategy is
	\begin{align*}
m(v) = \begin{cases}
                       [a_{k+1},a_k]\text{ where } a_{k+1}<v\leq a_k & \text{if }v>0, \\
                        \{0\} & \text{if }v=0.
                    \end{cases}\end{align*}
	When the monopolist receives message $m_k$, he charges $a_{k+1}$ thereby selling to that entire segment. Relative to no-personalized-pricing, this equilibrium strictly improves consumer surplus for all $v$ in $\left(0,1/2\right]$, and leaves consumer surplus unchanged for all other types.
\end{observation}

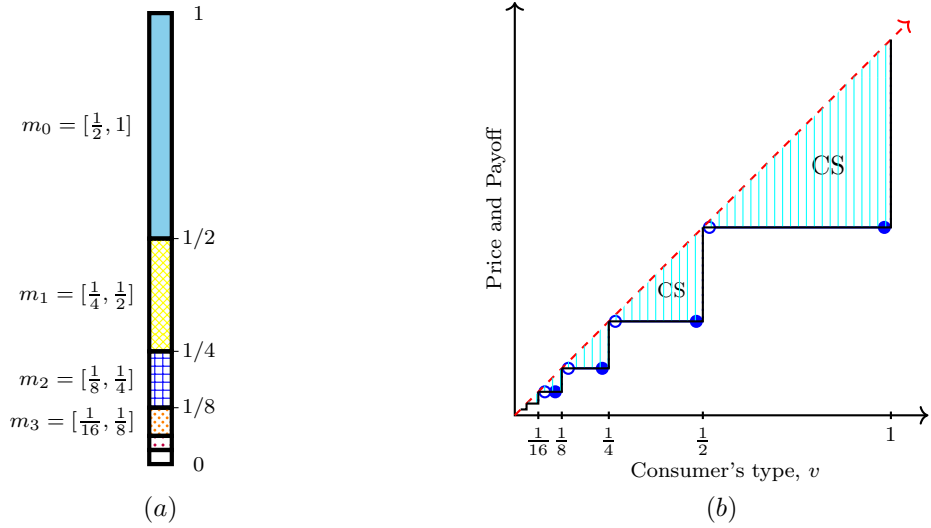
\begin{figure}[t]
\centering
{\scriptsize
	\begin{tikzpicture}[scale=0.3]
		\draw[ultra thick, pattern=north west lines, fill=SkyBlue] (0,12) rectangle (1,22);
		\node at (-3.2, 17){$m_0=[\frac{1}{2},1]$};
		\draw (1, 12)--(1.3,12);
		\node  at (2.2, 22) {$1$};			
		\node  at (2.2, 2) {$0$};
		\node  at (2.2, 12) {$1/2$};
		\draw (1, 7)--(1.3,7);
		\node  at (2.2, 7) {$1/4$};
		\draw (1, 4.5)--(1.3,4.5);
		\node  at (2.2, 4.5) {$1/8$};
		\draw[ultra thick, pattern=crosshatch, pattern color=yellow] (0,7) rectangle (1,12);
		\node at (-3.2, 9.5){$m_1=[\frac{1}{4},\frac{1}{2}]$};
		\draw[ultra thick, pattern=grid, pattern color=blue] (0,4.5) rectangle (1,7);
		\node at (-3.2, 5.75){$m_2=[\frac{1}{8},\frac{1}{4}]$};
		\draw[ultra thick, pattern=crosshatch dots, pattern color=orange] (0,3.25) rectangle (1,4.5);
		\node at (-3.4, 3.875){$m_3=[\frac{1}{16},\frac{1}{8}]$};
		\draw[ultra thick, pattern=dots, pattern color=purple] (0,2.625) rectangle (1,3.25);
		\draw[ultra thick] (0,2) rectangle (1,2.625);
		\node at (0.5,0) {\footnotesize $(a)$};
	\end{tikzpicture}\hspace{1.3in}
	\begin{tikzpicture}[scale=0.5]
		 \draw[<->, thick] (11,0) -- (0,0) -- (0,11) ;
		 \draw (5.5,-1.5) node{{ Consumer's type, $v$}};
		\node at (5.5,-2.5) {\footnotesize $(b)$};
		 \draw (-.5,5.5) node[rotate=90]{ Price and Payoff};
		 \draw (8.333,6.6667) node{\footnotesize CS};	 		 
		 \draw (4.1667,3.3333) node{CS};
		\draw[thick] (10, -0.1) node[below]{$1$} -- (10,0.1);
		\draw[thick] (5, -0.1) node[below]{$\frac{1}{2}$} -- (5,0.1);
		\draw[thick] (2.5, -0.1) node[below]{$\frac{1}{4}$} -- (2.5,0.1);
		\draw[thick] (1.25, -0.1) node[below]{$\frac{1}{8}$} -- (1.25,0.1);
		\draw[thick] (0.625, -0.1) node[below]{$\frac{1}{16}$} -- (0.625,0.1);
		\draw[thick, blue,Circle-{Circle[open]}] (10, 5) -- (5,5);
		\draw[thick, blue,Circle-{Circle[open]}] (5, 2.5) -- (2.5,2.5);
		\draw[thick, blue,Circle-{Circle[open]}] (2.5, 1.25) -- (1.25,1.25);
		\draw[thick, blue, Circle-{Circle[open]}] (1.25, 0.625) -- (0.625,0.625);
		\draw[thick, red, dashed, ->] (0,0)--(10.5,10.5);
		\draw[thick, blue, dotted] (10,5)--(10,10);
		\draw[thick, blue, dotted] (5,5)--(5,2.5);
		\draw[thick, blue, dotted] (2.5,2.5)--(2.5,1.25);
		\draw[thick, blue, dotted] (1.25,1.25)--(1.25,0.625);
		\draw[thick, blue, dotted] (0.625,0.625)--(0.625,0.3125);
		\draw[thick, pattern=vertical lines, pattern color=cyan] (10,10)--(10,5)--(5,5);
		\draw[thick, pattern=vertical lines, pattern color=cyan] (2.5,2.5)--(5,2.5)--(5,5);
		\draw[thick, pattern=vertical lines, pattern color=cyan] (2.5,2.5)--(2.5,1.25)--(1.25,1.25);
		\draw[thick, pattern=vertical lines, pattern color=cyan] (1.25,1.25)--(1.25,0.625)--(0.625,0.625);
		\draw[thick, pattern=vertical lines, pattern color=cyan] (0.625,0.625)--(0.625,0.3125)--(0.3125,0.3125);
		\draw[thick, pattern=vertical lines, pattern color=cyan] (0.3125,0.3125)--(0.3125,0.15625)--(0.15625,0.15625);		
	\end{tikzpicture}	
	}
\vspace{-0.1in}	
	\caption{(a) illustrates Zeno's Partition. (b) illustrates prices and payoffs: for each consumer-type $v$, the step-function shows the equilibrium price that is charged and the dashed $45^{\circ}$ line shows the payoff from consumption. The shaded region illustrates the consumer surplus achieved by Zeno's Partition.}\label{Figure-Zeno}
\end{figure}

In this equilibrium, the highest market segment is composed of types in $\left(\frac{1}{2},1\right]$, all of which send the message $m_0\equiv \left[\frac{1}{2},1\right]$; the next highest market segment comprises types in $\left(\frac{1}{4},\frac{1}{2}\right]$, all of which send the message $m_1\equiv \left[\frac{1}{4},\frac{1}{2}\right]$, and so on and so forth. We depict this partition in \autoref{Figure-Zeno}. Once the monopolist receives any message corresponding to each market segment, he believes that the consumer's value is uniformly distributed on it. His optimal strategy then is to price at the bottom of the segment. Therefore, trade occurs with probability $1$, with each higher consumer type capturing some surplus.

This equilibrium generates an ex ante consumer surplus of $\frac{1}{6}$ and producer surplus of $\frac{1}{3}$, each of which is higher than what is achieved without personalized pricing. All types in $(1/2,1]$ receive the same price that they would have if personalized pricing were infeasible, and almost every other type is strictly better off. Thus, personalized pricing improves the monopolist's profit, strictly improves the surplus for some consumer types without making any worse off.

How is Zeno's Partition supportable as an equilibrium? First, let us see what deters consumers from using messages that are not in Zeno's Partition. If the monopolist sees such a message, his off-path beliefs ascribe probability $1$ to the highest type that could send such a message, which leads him to charge a price equal to that type. Such beliefs ensure that these off-path messages are not profitable deviations. How about deviations to other on-path messages? For every $v$ in $(a_{k+1},a_k)$, there exists only one on-path message that it can send, and for every $v$ on the boundaries of such messages, our strategy profile prescribes the message that results in the lower price. Thus, there are no profitable deviations for any consumer type. Finally, by construction, the monopolist is always setting a price that is optimal given its beliefs.

It is useful to understand why we do not see unraveling. In many disclosure models, the sender strictly prefers to induce the receiver to have higher (or lower) beliefs in the sense of first-order stochastic dominance. Unraveling emerges as the unique equilibrium outcome as extreme types of the sender have a motive to separate from pools. By contrast, in our setting, there exist many pairs of beliefs $(\mu,\hat\mu)$ that are ranked by FOSD such that the sender is indifferent between inducing $\mu$ and $\hat\mu$ because they result in the receiver taking the same action. For example, the monopolist charges the same price when he ascribes probability $1$ to type $\{1/2\}$ as he does when his beliefs are $U[1/2,1]$. Thus, higher types may be pooled with a lower type without giving that lower type an incentive to separate. 

Zeno's Partition isn't the only equilibrium of this example, but in this case, it maximizes ex ante consumer surplus. To demonstrate why this is the case, we prove in \Cref{Section-MonopolistOptimalSegmentation} that for a one-dimensional type space, for every equilibrium, there exists an interim payoff-equivalent equilibrium in which trade occurs with probability $1$ and types segment into partitions. Thus, it is without loss of generality to consider only those equilibria that are fully efficient and partitional. Now we can use a heuristic argument to show why Zeno's Partition is optimal when types are uniformly distributed. 

Because consumers always purchase in a fully efficient equilibrium, maximizing consumer surplus is equivalent to minimizing the average price. For a monopolist to price at the bottom of an interval $[a,b]$ when $v$ is uniformly distributed between $a$ and $b$, it must be that $a\geq b/2$. Suppose that the consumer-optimal equilibrium involves types from $[\lambda,1]$ forming the highest segment; by the logic of the previous sentence, $\lambda$ is at least $1/2$. The monopolist charges a price of $\lambda$ to that segment, and thus, its contribution to the ex ante expected price is $(1-\lambda)\lambda$. The remaining population, $[0,\lambda]$, amounts to a $\lambda$-rescaling of the original problem, and so the consumer-optimal equilibrium after removing that highest segment involves replicating the same segmentation on a smaller scale. Thus, the consumer-optimal segmentation can be framed as a recursive problem where ${P}(\bar{v})$ is the lowest expected price generated by a partition when types are uniformly distributed on the interval $[0,\bar{v}]$: 
\begin{align*}
	{P}(1) &=\min_{\lambda\geq \frac{1}{2}}\, (1-\lambda)\lambda +\lambda {P}(\lambda)=\min_{\lambda\geq \frac{1}{2}}\, (1-\lambda)\lambda +\lambda^2 {P}(1)=\min_{\lambda\geq \frac{1}{2}}\, \frac{(1-\lambda)\lambda}{1-\lambda^2}=\frac{1}{3},
\end{align*}
where the first equality frames the problem recursively, the second is from ${P}(\lambda)$ being a re-scaled version of the original problem, and the remainder is algebra. Because Zeno's Partition induces an expected price of $1/3$, no segmentation can generate higher consumer surplus.

\subsection{Bertrand Competition with Horizontal Differentiation}\label{Section-ExampleCompetition}

The analysis in \Cref{Section-ExampleMonopolist} shows that in a monopolistic market, a consumer never benefits from using simple evidence but can obtain Pareto gains if the evidence structure is rich. Here we show that if there is market competition, even simple evidence suffices to generate Pareto gains in consumer surplus. We illustrate this effect in an  example of Bertrand competition with horizontal differentiation. 

Two firms, $L$ and $R$, compete to sell to a consumer who must purchase one unit of the good from either firm. Firm $L$ is located at the point $\ell_L=-1$, firm $R$ at the point $\ell_R=1$. The consumer's location, $\ell$, is drawn uniformly from $[-1,1]$. The consumer knows her location but the firms do not. If the consumer purchases the good from firm $i$, then she pays the price $p_i$ that is set by firm $i$, as well as a linear transportation cost $|\ell_i-\ell|$. 
If there were no voluntary disclosure, then each firm would set a price of $2$ in equilibrium. The consumer would then buy the good from the closer firm and incur a total expenditure of $2+\min\{1+\ell,1-\ell\}$. 

Let us describe what can happen with simple evidence. Now the consumer at location $\ell$ can disclose one of two messages to each firm privately before prices are set: either she can send a message of $\{\ell\}$, which fully reveals her location, or a message of $[-1,1]$, which fully conceals it. \Cref{Figure-ExampleSimpleCompetition} illustrates an equilibrium where the consumer selectively discloses evidence to amplify price competition and reduce her total expenditure. 
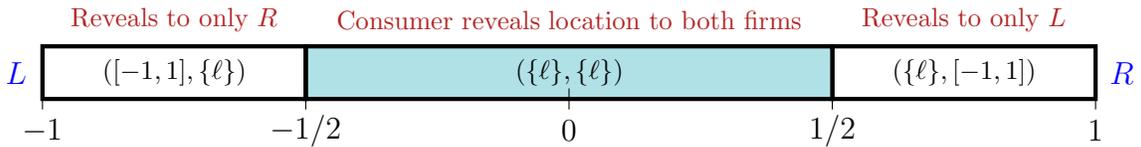
\begin{figure}[b]
\centering
{	\begin{tikzpicture}[scale=0.7]
		\draw[ultra thick] (15,0) rectangle (20,1); 
		\draw[ultra thick, pattern=north west lines, fill=PowderBlue] (5,0) rectangle (15,1); 
		\draw[ultra thick] (0,0) rectangle (5,1); 
		\draw (10,-0.2)--(10,0.2);
		\node at (10,-0.6) {$0$};
		\draw (15,-0.2)--(15,0.2);
		\node at (15,-0.6) {$1/2$};
		\draw (5,-0.2)--(5,0.2);
		\node at (5,-0.6) {$-1/2$};
		\draw (0,-0.2)--(0,0.2);
		\node at (0,-0.6) {$-1$};
		\draw (20,-0.2)--(20,0.2);
		\node at (20,-0.6) {$1$};
 		\node[blue] at (-0.5,0.5) {$L$};
 		\node[blue] at (20.5,0.5) {$R$};
 		\node at (2.5,0.5){{\footnotesize $([-1,1],\{\ell\})$}};
 		\node at (17.5,0.5){{\footnotesize $(\{\ell\},[-1,1])$}};
 		\node at (10,0.5){{\footnotesize $(\{\ell\},\{\ell\})$}};
 		\node[FireBrick] at (10,1.5){{\footnotesize Consumer reveals location to both firms}};
 		\node[FireBrick] at (2.5,1.5){{\footnotesize Reveals to only $R$}};
 		\node[FireBrick] at (17.5,1.5){{\footnotesize Reveals to only $L$}};
	\end{tikzpicture}	
	}
\vspace{-0.1in}
\caption{The figure shows disclosure strategies for every type. Centrally located types fully reveal location to both firms. Extreme types reveal location only to the distant firm and conceal it from the closer firm.}\label{Figure-ExampleSimpleCompetition}
\end{figure}

In this construction, the consumer conceals her location when she has a strong preferences for the product of one of the firms. Thus, if her location is in $[1/2,1]$, she reveals her location to the distant firm $L$ but conceals it from $R$; she plays a symmetric strategy if her location is in $[-1,-1/2]$. Only if her location is in $(-1/2,1/2)$ does she reveal her location to both firms. 

Why does this disclosure strategy lower market prices? Suppose that the consumer is located in $[1/2,1]$, and hence she discloses her location to firm $L$ but not to firm $R$. Firm $R$ infers from this non-disclosure that the consumer must be located in $[1/2,1]$, but does not know where in that interval. Firm $L$, on the other hand, knows the consumer's location, and to poach the consumer's business, it has to significantly reduce its price. In equilibrium, this powerful force pushes firm $L$ to lower its price all the way to its marginal cost of $0$. Anticipating this outside option for the consumer, firm $R$'s solves for its optimal local monopoly price, which is $1$. At these prices, the consumer purchases from firm $R$, and therefore has a total expenditure of $2-\ell$, which is strictly smaller than the total expenditure without personalized pricing.

What about all those interior types in $(-1/2,1/2)$ that are fully revealing themselves to each firm? All of these interior types are unraveling to both firms, but in this competitive market, this is an advantage: each obtains a price of $0$ from the distant firm and a price from the closer firm that makes it just indifferent. In equilibrium, each incurs a total expenditure of $\max\{1+\ell,1-\ell\}$, which once again is below that without personalized pricing. 

We see that consumer control, even through simple evidence, can be a powerful force for consumer gains. The idea is simple: by disclosing her preferences to distant firms, the consumer motivates them to lower prices. One's home firm then knows that even though the consumer has strong preferences for its product, she can purchase from another firm at a very low price. As we show in \Cref{Section-Competition}, this force manifests more broadly: it applies when there are more than two firms, and for both general distributions and general forms of product differentiation. We also show there that if the consumer has access to rich evidence, she can improve upon this segmentation using a Zeno-like construction just as we did in the monopolist example.

\section{Voluntary Disclosure to a Monopolist}\label{Section-MonopolistModel}
\subsection{Environment}\label{Subsection-Environment}
\paragraph{The Pricing Problem.}
A monopolist (``he'') sells a good to a single consumer (``she''), who demands a single unit. The consumer's type, denoted by $t$, is drawn according to a measure $\mu$ whose support is ${T}$. The type space $T$ is a convex and compact subset of a finite-dimensional Euclidean space, $\Re^k$. Each of these $k$ dimensions of a consumer's type reflect attributes that affect her valuation for the good according to $v: T\rightarrow \Re$. Payoffs are quasilinear: if the consumer purchases the good from the monopolist at price $p$ when her type is $t$, her payoff is $v(t)-p$ and the monopolist's payoff is $p$; otherwise, each player receives a payoff of $0$. We denote by $F$ the induced CDF over valuations; in other words, $F(\tilde v)\equiv\mu(\{t\in T:v(t)\leq \tilde{v})$. We denote by $\lowv$ and $\highv$ the highest and lowest valuations in the support. 
We simplify exposition by assuming that $F$ is continuous and $F(\lowv)=0$.\footnote{An equivalent assumption is that $\mu(\{t\in T:v({t})=v'\})=0$ for every $v'$ in the range of $v(\cdot)$. Conditions that guarantee this property are that $\mu$ is absolutely continuous with respect to the Lebesgue measure, and $v(\cdot)$ is strictly monotone in each dimension. }
  
Throughout our analysis, we assume that $v(t)$ is non-negative for every type $t\in \mathcal T$ and is quasiconvex.
A special leading case is where each dimension of $t$ is a consumer characteristic (e.g. income) and $v(t)$ is linear; in this case, $v(t)=\sum_{i=1}^k \beta_i t_i$ where $\beta_i$ is the coefficient on characteristic $i$. We order types based on their valuations: we say that $t\succeq t'$ if $v(t)\geq v(t')$, and we define $\succ$ and $\sim$ equivalently. When $t\succeq t'$, we refer to $t$ as being a \emph{higher} type.

Were communication infeasible, this pricing problem has a simple solution: the monopolist sets a price $p$ that maximizes $p(1-F(p))$. Let $\optprice$ denote the (lowest) optimal price for the monopolist. The consumer's interim payoff is then no more than $\max\{v(t)-\optprice,0\}$.

\paragraph{The Disclosure Game.} We append a disclosure game to this pricing problem. 
After observing her type, the consumer chooses a message $m$ from the set of messages available to her. 
The set of \emph{all} feasible messages is $\messages^{\mathcal F}\equiv \left\{M\subseteq T:M\text{ is closed and convex}\right\}$, and we interpret a message $M$ in $\messages^{\mathcal F}$ as meaning \emph{``My type is in the set $M$.''} 
When a consumer's type is $t$, the set of messages that she can send is $\messages(t)\subseteq \messages^{\mathcal F}$. We focus attention on the following two different forms of disclosure:
\begin{itemize}[noitemsep]
    \item the evidence structure is \textbf{simple} if for every $t$, $\messages(t)=\{ T,\{t\}\}$.
    \item the evidence structure is \textbf{rich} if for every $t$, $\messages(t)=\{M \in \messages^{\mathcal F}:t\in M\}$.
\end{itemize}
In both simple and rich evidence structures, the consumer has access to hard information about her type. In a simple evidence structure, the consumer can either disclose a ``certificate'' that fully reveals her type or says nothing at all. By contrast, in a rich evidence structure, the consumer can verifiably disclose true statements about her type without being compelled to reveal everything.
The assumption that messages are convex sets implies that if types $t$ and $t'$ can disclose some common evidence, then so can any intermediate type $t'' = \rho t+(1-\rho)t'$ (for $\rho\in (0,1)$).

\paragraph{Timeline and Equilibrium Concept.} First, the consumer observes her type $t$ and chooses a message $M$ from $\messages(t)$. The monopolist then observes the message and chooses a price $p\geq 0$. The consumer then chooses whether to purchase the good. We study Perfect Bayesian Equilibria (henceforth PBE) of this game. In other words, the seller's belief system is updated via Bayes Rule whenever possible, and both the consumer and the seller behave sequentially rationally.\footnote{If the seller receives message $m$, his beliefs (both on- and off-path) must have a support that is contained in $\{t\in T:m\in \messages(t)\}$.} For expository convenience, we assume that a consumer always breaks her indifference in favor of purchasing the good.

\subsection{Simple Evidence Does Not Help Consumers}\label{Section-SimpleEvidence}

Here, we show that when trading with a monopolist, consumers do not benefit from personalized pricing if the evidence structure is simple relative to a benchmark in which personalized pricing is impossible.
Recall that the interim payoff of each type $t$ without personalized pricing is $\max\{v(t)-p^*,0\}$ where $p^*$ is the monopolist's optimal price. Generalizing the argument of \Cref{Section-Example}, we show that there are equilibria with simple evidence that make all consumer types worse off and there are no equilibrium in which \emph{any} type is strictly better off. 

That consumers may be worse off is easily illustrated through the equilibrium in which the consumer fully reveals her type with probability $1$ and the monopolist charges a price of $v(t)$. Off-path, the seller's beliefs are \emph{maximally skeptical} in that he believes that the consumer's type has the highest possible valuation with probability $1$. The monopolist thus extracts all surplus, and so consumers are clearly worse off. 

But there are also partially revealing equilibria in which only those types below a cutoff are revealed. For example, there exists an equilibrium in which all types $t$ where $v(t)\geq p^*$ send the message $T$ and only types below $p^*$ reveal. This equilibrium results in payoffs for the consumer identical to that without personalized pricing. This is the consumer-optimal equilibrium. 

\begin{proposition}\label{Proposition-Certification}
With simple evidence, across all equilibria, the consumer's interim payoff is bounded above by $\max\{v(t)-p^*,0\}$.
\end{proposition}

Thus, the consumer gains {nothing}, \emph{ex ante} and \emph{ex interim}, from the ability to disclose her type using simple evidence. If one takes the model of simple evidence as a stylized representation of track / do-not-track regulations, our analysis implies that this form of consumer protection does not benefit consumers in a monopolistic environment relative to a benchmark that prohibits personalized pricing. Instead, richer forms of verifiable disclosure are needed.

\subsection{A Pareto-Improving Segmentation with Rich Evidence}\label{Section-GreedyPareto}

We develop a segmentation that generalizes that of \Cref{Section-Example}. 
Each segment is constructed so that  the monopolist's optimal price sells to all consumer types in that segment. Consumers would profit if they could deviate ``downwards'' to a lower  segment; our construction guarantees that this is impossible. Our construction is ``greedy'' insofar as we start with the highest segment and make each as large as possible without accounting for its effect on subsequent segments. 

To define the segmentation strategy, consider a sequence of prices  $\{p_s\}_{s=0,1,2,\ldots,S}$ where $S\leq\infty$, $p_0=\highv$, and for every $s$ where $p_{s-1}>\lowv$, $p_s$ is the (lowest) maximizer of $p_s(F(p_{s-1})-F(p_s))$. If $p_{s'}=\lowv$ for some $s'$, then we halt the algorithm and set $S=s'$; otherwise, $S=\infty$ and $p_\infty=\lowv$.
We use these prices to construct sets of types, $(M_s)_{s=1,2,\ldots,S}\cup M_\infty$:
\begin{align*}
    M_s&\equiv\{t\in T:v(t)\leq p_{s-1}\}.\\
    M_\infty&\equiv\{t\in T:v(t)=\lowv\}
\end{align*}
Because $v$ is quasiconvex and $T$ is convex, $M_s$ is a convex set for every $s=0,1,2,\ldots,S$, and therefore $M_s$ is a feasible message. These messages segment the market. 
\begin{proposition}\label{Proposition-MonopolistConsumerSegmentation}
    With rich evidence, there exists a Pareto-improving equilibrium in which a consumer's reporting strategy is  
	\begin{align*}
M^*(t) = \begin{cases}
                       M_s & \text{if }p_{s}<v(t)\leq p_{s-1}, \\
                        M_\infty & \text{if }t\in M_\infty.
                    \end{cases}\end{align*}    
When receiving an equilibrium disclosure of the form $M_s$, the seller charges a price of $p_s$ and sells to all types that send that message.  
\end{proposition}

The segmentation described above generalizes the ``Zeno Partition'' constructed in \Cref{Section-Example}. The highest market segment consists of those consumer types whose valuations strictly exceed the monopolist's optimal posted price, $p_1=p^*$; these are the types who send message $M_1$. The next highest market segment comprises those whose valuations exceed the optimal posted price, $p_2$, for the \emph{truncated distribution} that excludes the highest market segment; they send message $M_2$. This iterative procedure continues either indefinitely (if $p_s>\lowv$ for every $s$) or halts once the monopolist has no incentive to exclude any type in the truncated distribution from trading. 

Notice that in this segmentation, disclosures aren't taken at face value.
Instead, the monopolist infers from receiving a message $M_s$ that the consumer would have preferred to send message $M_{s+1}$ but couldn't, and so her valuation must be in $(p_s,p_{s-1}]$. 
Notice also that the market segmentation is constructed so that given these beliefs about the consumer's valuation, the monopolist has no incentive to charge a price that excludes any type. In fact, this constraint for the monopolist binds in our \emph{greedy segmentation} in that if types below $p_s$ were included, the seller's optimal price would exclude those types.

This equilibrium segmentation is fully efficient---trade occurs with probability $1$---and improves consumer surplus relative to the benchmark without personalized pricing. Consumer types in the highest market segment face the same price that they would without personalized pricing, but now consumers in other market segments can also purchase at prices that are (generically) below their willingness to pay. Thus, the segmentation is a Pareto improvement. One feature of the segmentation that is attractive is its simplicity: all that consumers have to disclose is information about their willingness to pay.

Finally, we note that this construction is robust to the possibility that the consumer does not have evidence with positive probability: there exists an equilibrium in this setting where if the consumer lacks evidence, she is charged a price of $p_1=p^*$, and all those with evidence behave as above. The consumer never gains from imitating those who lack evidence. 

\subsection{Optimal Equilibrium Segmentation}\label{Section-MonopolistOptimalSegmentation}

In this section, we explore conditions under which the disclosure strategy above is the ex ante optimal equilibrium for consumers. There are two reasons that this equilibrium may not maximize ex ante consumer surplus among equilibria. The first is that it ignores multidimensionality: even if two types have the same valuation, it may be optimal to separate them. The second reason is that even in a one-dimensional world, packing types greedily need not maximize consumer surplus. To explore this issue, we restrict attention to a one-dimensional model. We prove that the consumer-optimal equilibrium uses a partitional structure. We show by an example that the greedy partition may be suboptimal for some distributions. We then prove that it is consumer-optimal for power distribution functions. 

Let the set of types $T$ be identical to the set of values $[\lowv,\highv]$ and $v(t)=t$. Recall that $\lowv\geq 0$ and that $F$ is an atomless CDF over valuations. The restriction that $M(t)$ is a closed convex set that includes $t$ implies that here, $M(t)=\{[a,b]\subseteq [\lowv,\highv]:a\leq t\leq b\}$; in other words, the set of all closed intervals that include $t$. Applied to this setting, \Cref{Proposition-MonopolistConsumerSegmentation} identifies an equilibrium segmentation of the form $\left\{[0,p_s]\right\}_{s=0,1,2,\ldots,S}$ where $p_s$ is the optimal price to when the distribution $F$ is truncated to $[0,p_{s-1}]$. Because only types in $(p_{s+1},p_{s}]$ send the message $[0,p_s]$, a payoff-equivalent segmentation is for a type $t$ to send the message $[p_{s+1},p_s]$ where $p_{s+1}<t\leq p_s$. This equilibrium is ``partitional'' in that types reveal the member of the partition to which they belong, and thus, these messages can be taken at ``face value''. We prove that for any equilibrium, there always exists a payoff-equivalent equilibrium that is partitional and involves the sale happening with probability $1$.  

Our characterization uses the following definitions. A PBE is efficient if trade occurs with probability $1$. A collection of sets $\partition$ is a \textbf{partition} of $[\lowv,\highv]$ if $\partition$ is a subset of $\messages^F$ such that $\bigcup_{m\in\partition} m=[\lowv,\highv]$ and for every distinct $m,m'$ in $\partition$, $m\bigcap m'$ is at most a singleton. One message $m$ dominates $m'$ (i.e. $m\succeq_{\mathcal M} m'$) if for every $t\in m$ and $t'\in m'$, $t\geq t'$; $\arg\min$ and $\arg\max$ over a set of messages refers to this partial order. Given a partition $\partition$, let $m^{\partition}(t)\equiv \arg\min_{\{m\in \partition:t\in m\}}m$. An equilibrium $\sigma$ is \textbf{partitional} if there exists a partition $\partition$ such that $m^\sigma(t)=m^{\partition}(t)$, and for every $m$ in $\partition$, $p^\sigma(m)=\min_{t\in m} t$. 

\begin{proposition}\label{Proposition-MonopolistUnidimensional}
Given any equilibrium $\sigma$, there exists an efficient partitional equilibrium $\tilde\sigma$ that is payoff-equivalent for almost every type.    
\end{proposition}

The implication of \Cref{Proposition-MonopolistUnidimensional} is that it suffices to look at partitional equilibria. The proof of \Cref{Proposition-MonopolistUnidimensional} proceeds in two steps. First, we show that it is without loss of generality to look at efficient equilibria: for any equilibrium in which there exists a type that is not purchasing the product, there exists an interim payoff-equivalent equilibrium in which that type fully reveals itself to the seller. Second, we show that for any efficient equilibrium, there exists a partitional equilibrium that is payoff-equivalent for almost every type. To prove this step, we show that in any efficient equilibrium, prices must be (weakly) decreasing in valuation because otherwise some type has a profitable deviation. 

How does the greedy segmentation compare to other partitional equilibria from the perspective of ex ante consumer surplus?\footnote{From an interim perspective, the greedy partition is Pareto efficient because any partition that differs from it must raise the lowest type in at least one segment, which increases the price in that segment.} A consideration that the greedy algorithm ignores is that it may benefit average prices to exclude some high types from a pool, making those types pay a higher price, and pool intermediate types with low types. We illustrate this below. 
\begin{example}\label{Example-GreedyFailure}
Suppose that the consumer's type is drawn from $\left\{1/3,2/3,1\right\}$ where $Pr(t=1)=1/6$, $Pr(t=2/3)=1/3+\epsilon$, and $Pr(t=1/3)=1/2-\epsilon$, 
where $\epsilon>0$ is small. The greedy construction sets the highest segment as $\{2/3,1\}$---because the seller's optimal posted price here would be $2/3$---and the 
next segment as $\{1/3\}$. This segmentation results in an average price of $\approx 1/2$. But a better segmentation for ex ante consumer surplus involves the high type 
perfectly separating as $\{1\}$, and the next highest segment being $\{1/3,2/3\}$. This segmentation reduces the average price to $4/9$. 
\end{example}

Generally, the optimal segmentation can be formulated as the solution to a constrained optimization problem over partitions that minimizes the average price subject to the constraint that the monopolist finds it optimal to price at the bottom of each segment. The greedy algorithm offers a simple program where that constraint binds in each segment and \Cref{Example-GreedyFailure} indicates that this may be sub-optimal. Identifying necessary and sufficient conditions on distributions when such constraints necessarily bind is challenging because it requires understanding in detail how sharply the monopolist's optimal price responds to truncating the distribution at different points. This exercise is difficult for distributions where we cannot solve for the optimal price in closed-form.\footnote{Without solving for the closed-form, we can verify that the greedy algorithm is optimal if (i) $F$ is convex, and (ii) the optimal price on an interval $[0,\tilde{v}]$, denoted by ${p}(\tilde{v})$, has a slope bounded above by $1$ and is weakly concave.} A class of distributions where a closed-form solution is available is that of power distributions; for this class, the greedy algorithm identifies the consumer-optimal segmentation. 

\begin{proposition}\label{Proposition-GreedyOptimalPower}
    Suppose that $[\lowv,\highv]=[0,1]$ and the cdf on valuations, $F(v)=v^k$ for $k>0$. Then the greedy segmentation is the consumer-optimal equilibrium segmentation. 
\end{proposition}


\subsection{Discussion}\label{Section-Discussion}

Our analysis of the monopolistic setting concludes that (i) the combination of voluntary disclosure and personalized pricing does not benefit consumers if evidence is simple (\Cref{Proposition-Certification}), but (ii) it generates an interim Pareto improvement if evidence is rich (\Cref{Proposition-MonopolistConsumerSegmentation}). Thus, consumers' control over data benefits them when they can choose not only \emph{whether} to communicate but also \emph{what} to communicate. 

We have adopted an interpretation that considers the gains that consumers are able to obtain with different technologies. An alternative way to view our results is as a description of consumer payoffs across different kinds of equilibria. Suppose that the set of messages available to type $t$, $\messages(t)$, is richer than the rich evidence structure --- for instance, the set of all (Borel) subsets that contain $t$. A corollary of \Cref{Proposition-Certification} is that in any equilibrium in which all equilibrium path messages are either fully revealing or fully concealing (i.e., $m^*(t)\in \{\{t\},T\}$), no consumer type is better off than without personalized pricing. Analogously, the equilibrium constructed in \Cref{Proposition-MonopolistConsumerSegmentation} remains an equilibrium in this setting.\footnote{In both of these cases, the only adaptation that would have to be made is for the monopolist respond to a larger set of off-path messages; in each case, it suffices if the monopolist attributes every off-path message to a type with the highest valuation that could have sent it.}

\section{How Disclosure Amplifies Competition}\label{Section-Competition}


In many settings, consumers do not interact with only one seller but instead face a competitive market in which firms are differentiated. In this section, we investigate the conditions under which voluntary disclosure and personalized pricing benefit consumers in competitive settings with differentiated products. Our general analysis allows for two or more firms and a general kind of product differentiation that encompasses both horizontal and vertical components. But for expositional clarity, we begin with the case of Bertrand duopoly with horizontally differentiated products, since this analysis elucidates all of the key economic forces.


 \Cref{Section-CompetitionEnvironment} describes the Bertrand duopoly setting with horizontal differentiation and \Cref{Section-CompetitionEquilibria} constructs equilibrium segmentations with simple and rich evidence. \Cref{Section-CompetitionWelfare} compares the consumer's payoffs with those of a benchmark setting without personalized pricing. These sections generalize the example of \Cref{Section-ExampleCompetition} in which the consumer's location is uniformly distributed. \Cref{Section-CompetitionMoreThanTwo} considers the general setting with $n\geq 2$ firms. 

\subsection{Bertrand Duopoly with Horizontal Differentiation}\label{Section-CompetitionEnvironment}

Two firms, $L$ and $R$, compete to sell to a single consumer who has unit demand. The type of the consumer is her \emph{location}, denoted by $t$, which is drawn according to measure $\mu$ (and cdf $F$) with support $T$. We assume that $T\equiv [-1,1]$ and that $F$ is atomless with a strictly positive and continuous density $f$ on its support. The firms $L$ and $R$ are located at the two end points, respectively $-1$ and $1$, and each firm $i$ sets a price $p_i\geq 0$. The consumer has a value $V$ for buying the good that is independent of her type $t$, and faces a ``transportation cost'' when purchasing from firm $i$ that is proportional to the distance between her location and that of the firm, $\ell_i$.
Thus, her payoff from buying the good from firm $i$ at a price of $p_i$ is $V -|t - \ell_i| - p_i$.
As is standard, we assume that $V$ is sufficiently large that in the equilibria we study below, all types of the consumer purchase the good and no type is excluded from the market.\footnote{See \cite{osborne1987equilibrium}, \cite{caplin1991aggregation}, \cite{bester1992}, and \cite{peitz1997differentiated}. For most of our analysis, it suffices for $V\geq 2$, so that a consumer is always willing to purchase the good from the most distant firm if that distant firm sets a price of $0$. } 

\paragraph{Disclosure.} After observing her type, the consumer chooses a message $M$ that is feasible and available for her to send to each of the firms. As before, the set of feasible messages is $\messages^{\mathcal{F}} \equiv \left\{[a,b]:-1\leq a \leq b \leq 1\right\}$ where a message $[a,b]$ is interpreted as ``\textit{my type is in the interval $[a,b]$}." When a consumer's
type is $t$, the set of messages that she can send is $\messages(t) \subseteq \messages^{\mathcal{F}}$. We study two disclosure technologies:
\begin{itemize}[noitemsep]
    \item \textbf{simple} evidence messages for each type $t$, $\messages(t) = \{ [-1,1], \{t\} \}$.
    \item \textbf{rich} evidence messages for each type $t$, $\messages(t) = \{ [a,b] : a \leq t \leq b \}$.
\end{itemize}

Each evidence technology is identical to its counterpart in the monopolistic model when the type space is unidimensional. The novelty here is that the consumer now sends two messages---$M_L$ to firm L and $M_R$ to firm R---and each message is privately observed by its recipient.  Both messages come from the same technology but are otherwise unrestricted. For example, a consumer of type $t$ can reveal her type by sending the message $\{t\}$ to one firm while concealing it from the other firm using the message $[-1,1]$. 

\paragraph{Timeline and Equilibrium Concept.} The consumer first observes her type $t$ and then chooses a pair of messages $(M_L, M_R)$, each from $\messages(t)$.\footnote{While we have treated $t$ as the consumer's location, our analysis is also compatible with a setting where all that a consumer observes is a signal with her posterior expected location, like \cite{armstrongzhou2019}, and chooses whether and how to disclose that expected location using simple or rich evidence.} 
Each firm $i$ privately observe its message $M_i$ and sets price $p_i \geq 0$; price-setting is simultaneous. The consumer then chooses which firm to purchase the good from, if any. 

We study Perfect Bayesian Equilibria of this game. As is well-known \citep{osborne1987equilibrium,caplin1991aggregation}, the price-setting game in Bertrand competition with horizontal differentiation may lack a pure-strategy equilibrium for general distributions. By contrast, we show constructively that pure-strategy equilibria always exist when this market setting is augmented with a disclosure game.

\subsection{Constructing Equilibria with Simple and Rich Evidence}\label{Section-CompetitionEquilibria}
This section constructs equilibria of the disclosure game with simple and rich disclosure technologies for any distribution of consumer types. In both cases, we use the following strategic logic. Each consumer reveals her type to the firm that is more distant from her, indicating that she is ``out of reach." This distant firm then competes heavily for her business by setting a low price, which in equilibrium equals $ 0 $. The firm who does not obtain a fully revealing message infers that the consumer is closer to his location. Based on that inference, this seller sets a profit-maximizing price subject to the consumer having the option to buy from the other seller at a price of $0$. We use the assumption that $V\geq 2$ to guarantee that the consumer weakly prefers purchasing the good from the distant firm at a price of $0$ to not purchasing it at all. We begin our analysis with a fully revealing equilibrium in both simple and rich evidence environments, and then show how to improve upon it. 

\begin{proposition}
    \label{Proposition-DuopolistSimpleUnraveling}
    There exists a fully revealing equilibrium in both simple and rich evidence games: every type of consumer $t$ sends the message $\{t\}$ to each firm, and purchases from the firm nearer to her at a price of $2|t|$.
\end{proposition}

The logic of \Cref{Proposition-DuopolistSimpleUnraveling} is straightforward. In an equilibrium where the consumer reveals her location to each firm,  both firms do not charge her strictly positive prices in equilibrium. Standard Bertrand logic implies that the distant firm must charge her a price of $0$ and the closer firm charges her the highest price that it can subject to the constraint that the consumer finds it incentive-compatible to purchase from the closer firm at that price.\footnote{Once types are revealed, these equilibrium prices necessarily coincide with those of \cite{thissevives1988}, where the consumer's type is common knowledge.} If the consumer deviates by sending a message $M$ that isn't a singleton to firm $i$, then firm $i$ believes that the consumer's type is the one in $M$ closest to $\ell_i$ and that the consumer has revealed her location to firm $j$. This equilibrium, thus, involves each seller holding skeptical beliefs that the consumer is as close as possible (given the message that is sent). 

This fully revealing equilibrium serves central types very well because they benefit from intense price competition. However, extreme types suffer from the firm closer to them being able to charge a high price. Ideally, types that are located close to firm $i$ would benefit from pooling with types more distant from firm $i$. The next result uses simple evidence to construct a partial pooling equilibrium that improves upon the fully revealing equilibrium for a strictly positive measure of types without making any type worse off. 

Our construction uses the following notation. Let $p^i_1$ be the lowest maximizer of $p\ell_i(F(\ell_i)-F(p\ell_i/2))$, and let $t^i_1\equiv{p^i_1\ell_i/2}$. To provide some intuition, $p^i_1$ is the (lowest) optimal price that firm $i$ charges if he has no information about the consumer's type and firm $j$ charges a price of $0$; in other words, this is firm $i$'s optimal \emph{local monopoly price} against an outside option where firm $j$ charges a price of $0$. At these prices, firm $i$ expects to sell to the consumer with probability $\ell_i(F(\ell_i)-F(p\ell_i/2))$, and $t^i_1$ is the most distant type from firm $i$ that still purchases from firm $i$. It is necessarily the case that $-1<t^L_1<0<t^R_1<1$. We use these types to describe our equilibrium.

\begin{proposition}\label{Proposition-DuopolistSimpleExistence}
    With simple evidence, there exists a partially pooling equilibrium in which the consumer's reporting strategy is
    		\begin{align*}
	\big(M^{\ast}_{L}(t), M^{\ast}_{R}(t) \big) = \begin{cases}
	\big( [-1,1], \{t\} \big) & \text{if } -1 \leq t \leq t^L_1, \\
	\big( \{t\}, \{t\} \big) & \text{if } t^L_1 < t < t^R_1, \\
	\big( \{t\}, [-1,1] \big) & \text{if } t^R_1 \leq t \leq 1,
	\end{cases}\end{align*}
and the prices charged by firm $i$ are
\begin{align*}
p_i^*(M) = \begin{cases}
\max\{2t\ell_i,0\} & \text{if } M=\{t\}, \\
p^i_1 & \text{otherwise.} 
\end{cases}\end{align*}  
In equilibrium, every consumer type purchases from the seller nearer to her.	
\end{proposition}

	
	

An intuition for \Cref{Proposition-DuopolistSimpleExistence} is as follows. If the consumer is centrally located---i.e., in $(t_L^1,t_R^1)$---she discloses her type (``track") to both firms. Such consumers then benefit from intense price competition, exactly as in the fully revealing equilibrium of \Cref{Proposition-DuopolistSimpleUnraveling}. If the consumer is not centrally located, she reveals her location to the firm farther from her but not to the nearer one. This private messaging strategy guarantees that the distant firm prices at zero and offers an attractive outside option. The firm that receives an uninformative (``don't track") message infers that the consumer is located sufficiently close but does not learn where. That firm then chooses an optimal local monopoly price given the outside-option price of zero. This pool of extreme types improves consumer welfare by guaranteeing that extreme consumer types can pool with moderate types, thereby decreasing type-contingent prices relative to the fully revealing equilibrium. \Cref{Figure-ExampleSimpleCompetition} in \Cref{Section-ExampleCompetition} depicts this segmentation for the case where the consumer's location is uniformly distributed.

One can do even better with rich evidence by using a segmentation that is analogous to the ``Zeno Partition'' constructed in \Cref{Section-GreedyPareto}. In this case, the central type $t=0$ obtains equilibrium prices of $0$ from each firm, and plays a role similar to the lowest type in the monopolistic setting. Accordingly, one sees a segmentation that goes from the extremes to the center, and becomes arbitrarily fine as one approaches the center. To develop notation for this argument, let us define a sequence of types $\{t_s^i\}_{s=0,1,2,\ldots}$ and prices and messages $\{p_s^i,M_s^i\}_{s=1,2,\ldots}$ where for every firm $i$ in $\{L,R\}$:
\begin{itemize}[noitemsep]
    \item $t_0^i=\ell_i$ and for every $s>0$, $t_s^i=p_s^i\ell_i/2$.
    \item $p_s^i$ is the lowest maximizer of $p\ell_i(F(t_{s-1}^i)-F(p\ell_i/2))$.
    \item $M_s^i\equiv\{t\in [-1,1]:t_s^i\ell_i\leq t\ell_i\leq t_{s-1}^i\ell_i\}$.
\end{itemize}
Let $p_\infty^i=0$ and let $M_\infty^i=\{0\}$. We have thus defined a sequences of cutoffs, prices, and messages where at every stage, we are constructing segments greedily: given a segment $M_s^i$, firm $i$ charges the price that is the optimal local monopoly price (assuming that the other firm charges a price of $0$), and at this price, firm $i$ services all consumer types in $M_s^i$. Because rich evidence allows consumers to disclose intervals directly, our disclosure strategy need not be asymmetric (unlike our analysis of the segmentation with simple evidence): a consumer of type $t$ can send the message $M_s^i$ that contains $t$ to both firms. We use this notation to prove our result below. 

\begin{proposition}
	\label{Proposition-DuopolistRichExistence}
	With rich evidence, there exists a segmentation equilibrium in which a consumer's reporting strategy is to send message $M^*(t)$ to both firms where 
	\begin{align*}
		M^*(t) = \begin{cases}
		 M_s^i  & \text{if }  t_s^i\ell_i < t\ell_i \leq t_{s-1}^i\ell_i \\
		 M_{\infty}^i  & \text{if } t = 0.  
	\end{cases}\end{align*}    
When receiving an equilibrium disclosure of the form $ M_s^i $, firm $i$ charges a price of $p_s^i$ and firm $j$ charges a price of $0$.
\end{proposition}

This equilibrium construction highlights the versatility of rich evidence disclosure. While the competitive environment differs from the monopolistic setting in many ways, the logic of the ``Zeno Partition" strategy follows in much the same way. Consumers with the highest willingness to pay for the good from firm $i$ are segmented together and send messages $M_1^i $. That message induces a price of $0$ by firm $j$ and given that outside option, firm $i$ charges a price that makes indifferent the marginal consumer type in $M_1^i$, who has the lowest willingness to pay for firm $i$'s product. Prices diminish as the consumer types become closer to the center. As such, the segmentation follows iteratively from both sides of $ 0 $ exactly as in ``Zeno."\footnote{For the uniform distribution, the construction mirrors that in \Cref{Section-Example} where $t_s^i=\ell_i(1/2)^s$. } We depict this segmentation strategy in \Cref{Figure-RichCompetition}.

\begin{figure}[h]
\centering
{	\begin{tikzpicture}[scale=0.7]
		\draw[ultra thick, fill=PeachPuff] (15,0) rectangle (20,1); 
		\draw[ultra thick, pattern=north west lines, pattern color=Orchid] (5,0) rectangle (7.5,1); 
		\draw[ultra thick, pattern=grid, fill=yellow] (7.5,0) rectangle (8.75,1); 
		\draw[ultra thick, pattern=north west lines, fill=gray] (8.75,0) rectangle (11.25,1); 
		\draw[ultra thick, pattern=grid, fill=SkyBlue] (11.25,0) rectangle (12.5,1);
		\draw[ultra thick, pattern=dots, pattern color =gray] (12.5,0) rectangle (15,1);		

		\draw[ultra thick] (0,0) rectangle (5,1); 
		\draw (10,-0.2)--(10,0.2);
		\node at (10,-0.6) {$0$};
		\draw (5,-0.2)--(5,0.2);
		\node at (5,-0.6) {$t_1^L$};
		\draw (7.5,-0.2)--(7.5,0.2);
		\node at (7.5,-0.6) {$t_2^L$};
		\draw (8.75,-0.2)--(8.75,0.2);
		\node at (8.75,-0.6) {$t_3^L$};
	    \draw (15,-0.2)--(15,0.2);
		\node at (15,-0.6) {$t_1^R$};		
	    \draw (12.5,-0.2)--(12.5,0.2);
		\node at (12.5,-0.6) {$t_2^R$};
        \draw (11.25,-0.2)--(11.25,0.2);
		\node at (11.25,-0.6) {$t_3^R$};		
		\draw (0,-0.2)--(0,0.2);
		\node at (0,-0.6) {$-1$};
		\draw (20,-0.2)--(20,0.2);
		\node at (20,-0.6) {$1$};
 		\node[blue] at (-0.5,0.5) {$L$};
 		\node[blue] at (20.5,0.5) {$R$};
 		\node at (2.5,0.5){{\footnotesize $M^L_1$}};
 		\node at (17.5,0.5){{\footnotesize $M^R_1$}};
 		\node at (6.25,0.5){{\footnotesize $M^L_2$}};
 		\node at (8.125,0.5){{\footnotesize $M^L_3$}}; 		
 		\node at (13.75,0.5){{\footnotesize $M^R_2$}};
 		\node at (11.875,0.5){{\footnotesize $M^R_3$}};		
 	\end{tikzpicture}	
	}
\vspace{-0.1in}
\caption{The figure shows a segmentation using rich evidence. The types in $(t_3^L,t_3^R)$ are partitioned into countably infinitely many segments, and hence these segments are omitted.}\label{Figure-RichCompetition}
\end{figure}
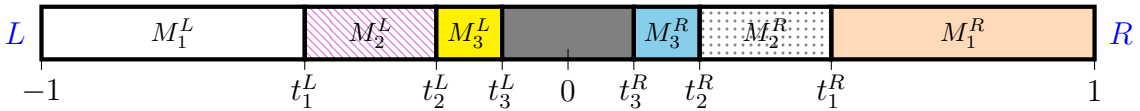

We have constructed equilibria with simple and rich evidence but we do not argue that these equilibria are consumer-optimal: an equilibrium segmentation (with either simple or rich evidence) may generate segments that induce each firm to randomize in its pricing strategy. Characterizing or bounding prices across mixed strategy equilibria across segments appears intractable.\footnote{Restricting attention to segmentations that generate pure-strategy equilibria, we conjecture that the equilibrium that we construct in \Cref{Proposition-DuopolistSimpleExistence} is consumer-optimal in the game with simple evidence; similarly, we conjecture the same regarding the equilibrium constructed in the game with rich evidence whenever the greedy algorithm yields an optimal segmentation.} Instead, we compare these equilibria to that of a benchmark model without personalized pricing and show that these equilibria generate strict interim Pareto gains.

\subsection{Benefits of Personalized Pricing in Competitive Markets}\label{Section-CompetitionWelfare}

The benchmark is the standard model of Bertrand pricing with horizontal differentiation: each firm $i$ sets a uniform price $p_i$ and the consumer buys from one of the firms. Unfortunately, this game may lack a pure-strategy equilibrium (in prices), and characterizing the mixed strategy equilibria is challenging. Accordingly, we impose a distributional assumption that is standard in this setting: we assume that $f$ is symmetric around $0$ and is strictly log-concave. This assumption guarantees the existence and uniqueness of a symmetric pure strategy equilibrium \citep{caplin1991aggregation}, and is compatible with many standard distributions \citep{bagnoli2005log}.

With this assumption, a symmetric pure-strategy equilibrium in this benchmark setting exists and involves each firm charging a price of $p^*$ where 
\begin{align*}
    p^*\equiv\arg\max_p pF\left(\frac{p^*-p}{2}\right)=\frac{2F(0)}{f(0)},
\end{align*}
where the first equality is firm $L$'s profit maximization problem, and the second comes from solving its first-order condition and substituting $p=p^*$.\footnote{As before, we assume that $V$ is sufficiently high that all consumers purchase at these prices. It suffices that $V>(f(0))^{-1}+1$.} Our welfare result compares this price to those of the equilibria constructed in \Cref{Section-CompetitionEquilibria}.
\begin{proposition}\label{Proposition-DuopolyWelfare}
    If $f$ is symmetric around $0$ and log-concave, then every type has a strictly higher payoff in the equilibria of the simple and rich evidence games constructed in \Cref{Proposition-DuopolistSimpleExistence,Proposition-DuopolistRichExistence} than in the benchmark setting without personalized pricing.
\end{proposition}
The logic of \Cref{Proposition-DuopolyWelfare} is that the price in the benchmark setting ($p^*$) is strictly higher than $p_1^i$, the price charged by firm $i$ to a consumer who conceals her type from firm $i$ in the equilibrium of the simple evidence game (\Cref{Proposition-DuopolistSimpleExistence}). The consumer must then be better off because this price ($p_1^i$) is strictly higher than all other equilibrium path prices both in this equilibrium and in the equilibrium that we construct in the rich evidence game. We illustrate the welfare gains from simple evidence in \Cref{Figure-DuopolyWelfare} for the case of the uniform distribution.

\begin{figure}[t]
\centering
{	\begin{tikzpicture}[scale=1]
	\draw (0,6)-- (0,0) --  (10,0) -- (10,6);
	\draw (0 cm,3pt) -- (0 cm,-3pt) node[anchor=north] {$-1$};
	\draw (10 cm,3pt) -- (10 cm,-3pt) node[anchor=north] {$1$};	
	\draw (5 cm,3pt) -- (5 cm,-3pt) node[anchor=north] {$0$};
	\draw (2.5 cm,3pt) -- (2.5 cm,-3pt) node[anchor=north] {$-\frac{1}{2}$};
	\draw (7.5 cm,3pt) -- (7.5 cm,-3pt) node[anchor=north] {$\frac{1}{2}$};	
	\draw (-1,3) node[rotate=90]{Price + Transport Cost};	
	\draw (5,-1.2) node{Consumer's Type / Location};
	\draw (5,5) node{{\footnotesize No Personalized Pricing}};
	\draw (5,1.5) node{{\footnotesize Simple Evidence}};
	\draw (-4pt,6) -- (3pt,6);
	\draw (-8pt,6) node {$3$};
	\draw (-4pt,4) -- (3pt,4);
	\draw (-8pt,4) node {$2$};
	\draw (-4pt,2) -- (3pt,2);
	\draw (-8pt,2) node {$1$};
{\draw[red] (0,4)--(5,6)--(10,4);}
{\draw[black,densely dashed] (0,4)--(5,2)--(10,4);}
{\draw[blue,ultra thick] (0,2)--(2.5,3)--(5,2)--(7.5,3)--(10,2);}
	\end{tikzpicture}}
\vspace{-0.1in}	
\caption{The figure compares the interim equilibrium cost (incl. price and transport cost) in the setting without personalized pricing with that of the equilibrium constructed in the simple evidence game (\Cref{Proposition-DuopolistSimpleExistence}) for uniformly distributed types. Simple evidence reduces the expected cost by $50\%$. }\label{Figure-DuopolyWelfare}	
	\end{figure}
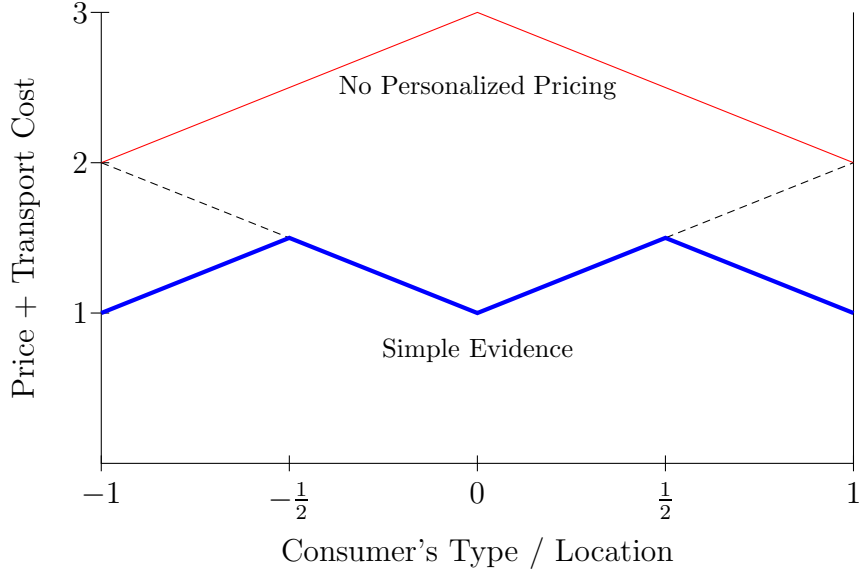

The ability for extreme types to pool is needed for these interim Pareto gains. These gains would not generally emerge if the consumer fully revealed her type: apart from the uniform distribution, any symmetric log-concave density has $f(0)>1/2$ and therefore, prices in the benchmark are strictly less than $2$. By contrast, in the fully revealing equilibrium, the extreme types pay a price of $2$. Thus, these interim Pareto gains emerge when consumers can disclose or conceal evidence about their types, and not necessarily when types are commonly known. 

\subsection{Competition Between More than Two Firms}\label{Section-CompetitionMoreThanTwo}

In this section, we show how our prior conclusions generalize to the case of multiple firms that produce differentiated products. Although the economic intuition for how disclosure amplifies competition remains the same, the arguments and notation are necessarily more involved to account for the higher dimensionality introduced by there being more firms. 

Suppose that the set of firms is $N\equiv\{1,\ldots,n\}$ where the number of firms, $n$, is at least $2$. The consumer has a valuation for each of these products, encoded in its type, $t\equiv (t_1,\ldots,t_n)$, where $t_i$ is the consumer's value for the good produced by firm $i$. We assume that $t$ is drawn according to distribution $\mu$ whose support is $T\equiv\left[\underline{t},\overline{t}\right]^n$, where $0<\underline{t}<\overline{t}$; for simplicity, we assume that $\mu$ has a strictly positive and continuous density $f$ on its support. We say that firm $i$ is type $t$'s \emph{favorite firm} if $t_i$ is weakly higher than $t_j$ for every firm $j$. 

We first construct a partial pooling equilibrium using simple evidence. We then construct an equilibrium with rich evidence. Finally, we compare the equilibrium prices under the two technologies to the model without personalized pricing. 

\paragraph{Simple Evidence:} For each type $t$, the set of messages available to the consumer  is $\messages(t)=\{T,\{t\}\}$. Messages are private and $M_i(t)$ denotes the message sent to firm $i$. It is useful to define a demand function for firm $i$ assuming that all other firms charge a price of $0$. For every non-negative price $p$, let
\begin{align*}
    Q^i(p)&\equiv\mu\left(\{t:t_i-p\geq \max_{j\neq i}t_j\}\right).
\end{align*}
The  demand $Q^i(p)$ is the probability that the consumer purchases from firm $i$ at a price of $p$ when all other firms charge $0$. Let $p^i_1$ denote the lowest maximizer of $pQ^i(p)$. The following result constructs an equilibrium with simple evidence.
\begin{proposition}\label{Proposition-OligopolySimpleExistence}
    With simple evidence, there exists an equilibrium in which the consumer's reporting strategy is
		\begin{align*}
	M^{\ast}_{i}(t)= &\begin{cases}
	\{t\} & \text{if } t_i-p^i_1< \max_{j\neq i} t_j, \\
	T & \text{otherwise.}
	\end{cases}
\end{align*}
{The prices charged by firm $i$ are}
\begin{align*}
p_i^*(M) = &\begin{cases}
\max\{0,t_i-\max_{j\neq i} t_j\} & \text{if } M=\{t\}, \\
p^i_1 & \text{otherwise.} 
\end{cases}\end{align*}  
In equilibrium, every consumer type purchases from its favorite firm.
\end{proposition}
This equilibrium generalizes the simple partitional form of \Cref{Proposition-DuopolistSimpleExistence} and \Cref{Figure-ExampleSimpleCompetition}. The consumer sends the non-discriminatory message $T$ only when she has strong preferences for the product of a particular firm. Following these messages, whenever firm $i$ receives the message $T$, it infers that the consumer has a strong preference for its product and that the consumer has sent a fully revealing message to every other firm. Anticipating that every other firm is then charging a price of $0$, firm $i$ then charges its optimal local monopoly price, which is $p_1^i$. When the consumer has mild preferences for the products of each firm, she sends the fully revealing message to each. The firms then compete for her business using standard Bertrand prices (with differentiated products). We omit a proof of this result as it is nearly identical to that of \Cref{Proposition-DuopolistSimpleExistence}.

\paragraph{Rich Evidence:}

We now construct an equilibrium using the rich evidence structure of \Cref{Section-MonopolistModel} where $\messages(t)$, the set of messages that the consumer can send when her type is $t$, is the set of all closed and convex subsets of $T$ that contain $t$. We partition $T$ on the basis of each consumer type's favorite firm, and the closest competitor that would like to poach that type's business. For each consumer type $t$, let $\alpha(t)$ denote type $t$'s favorite firm and $\beta(t)$ denote type $t$'s favorite among the remaining firms, breaking ties in favor of the firm that has the lowest index in each case.\footnote{Formally, treating $\min$ as the lowest element of a set, $\alpha(t)=\min {\{i\in \{1,\ldots,n\}:t_i \geq t_j\text{ for all }j\}}$, and $\beta(t)=\min {\{i\in \{1,\ldots,n\}\backslash\{\alpha(t)\}:t_i \geq t_j\text{ for all }j\neq \alpha(t)\}}$. These tie-breaking rules are for completeness, but any tie-breaking rule that preserves $\alpha(t)\neq \beta(t)$ suffices.} It is helpful to define a localized demand function: for every non-negative price $p$ and every (Borel) set $A$,
\begin{align*}
    Q^i(p,A)\equiv \mu \left(\{t\in A:t_i - p\geq \max_{j\neq i} t_j\}\right).
\end{align*}
This term is the probability that the consumer purchases from firm $i$ at price $p$ when the consumer's type is in $A$ and every other firm is charging a price of $0$.

For each firm $i$ and competitor $j$, define a sequence of prices $\{p^{ij}_s,\}_{s=0,1,2,\ldots}$ and sets of types $\{M^{ij}_s\}_{s=0,1,2,\ldots}$ such that $p^{ij}_0=\overline t-\underline t$ and 
\begin{align*}
    M^{ij}_s &\equiv \left\{t\in T:i=\alpha(t),j=\beta(t),t_i-t_j\leq p^{ij}_{s-1}\right\},\\
    p^{ij}_{s}&\text{ is the smallest maximizer of } pQ(p,M^{ij}_s).
\end{align*}
The set $M^{ij}_s$ denotes all consumer types for which $i$ is the favorite, $j$ is the second favorite, and each type is willing to pay as much as $p^{ij}_{s-1}$ to obtain the product from firm $i$. As we show in the Appendix, this is a convex set. The price $p^{ij}_s$ is the local monopoly price charged by firm $i$ when it knows that the consumer's type is in $M^{ij}_s$. Denoting the set of types for whom firms $i$ and $j$ are two equally favorite firms by $M^{ij}_\infty \equiv \{t\in T:i=\alpha(t),j=\beta(t),t_i=t_j\}$ so that $p^{ij}_\infty =0$, we can now state the following result.
\begin{proposition}\label{Proposition-OligopolyRichExistence}
    With rich evidence, there exists an equilibrium in which a consumer's reporting strategy is 
\begin{align*}    
M^*(t) = \begin{cases}
                       M^{ij}_s & \text{if }t\in M^{ij}_s\text{ and }t_i-t_j>p^{ij}_s, \\
                        M^{ij}_\infty & \text{if }t\in M^{ij}_\infty.
                    \end{cases}\end{align*}    
When receiving an equilibrium disclosure of the form $M^{ij}_s$, firm $i$ charges a price of $p^{ij}_s$, and all other firms charge a price of $0$. The consumer always purchases the good from her favorite firm. 
\end{proposition}
The construction here merges that of \Cref{Proposition-MonopolistConsumerSegmentation} and \Cref{Proposition-DuopolistRichExistence}: consumer types send messages that reveal their two favorite firms as well as a maximal willingness to pay for the good produced by their favorite firm. All non-favorites charge a price of $0$ and the favorite firm charges a price corresponding to the local monopoly price.

\paragraph{Benefits of Personalized Pricing:} We show that the equilibria constructed in \Cref{Proposition-OligopolySimpleExistence,Proposition-OligopolyRichExistence} improve consumer surplus relative to the benchmark of no personalized pricing. For the benchmark model, to guarantee full market coverage, we assume that the consumer has to purchase the product from one firm. We say that the distribution is symmetric if whenever $t'$ is a permutation of $t$, $f(t')=f(t)$. 
\begin{proposition}\label{Proposition-OligopolyWelfare}
If $f$ is symmetric and log-concave, then every type has a strictly higher payoff in the equilibria constructed in \Cref{Proposition-OligopolySimpleExistence,Proposition-OligopolyRichExistence} than in the benchmark setting without personalized pricing.   
\end{proposition}
Thus, we see that our conclusions from Bertrand duopoly with horizontal differentiation apply broadly: voluntary disclosure with either simple or rich evidence can amplify competition and generate gains for consumers. Our analysis thus highlights how consumer control through a simple track / do-not-track technology is sufficient in competitive markets.

\section{Conclusion}
\label{Section-Conclusion}

As the digital economy matures, policymakers and industry leaders are working to establish norms and regulations to govern data ownership and transmission. In light of the privacy and distributional concerns that this issue raises, we set out to study the question: \emph{do consumers benefit from personalized pricing when they have control over their data?} We frame and answer this question using the language of voluntary disclosure, building on a rich theoretical literature on evidence and hard information. 

Our initial instinct was that voluntary disclosure would not help. As the market draws inferences based on information that is \emph{not} disclosed, giving consumers the ability to separate themselves would seem to be self-defeating. To put it differently, if the market necessarily unravels as in \cite{grossman:81-JLE} and \cite{milgrom1981good}, consumers retain no surplus and may be worse off with personalized pricing.
We show that this conclusion is incorrect because it omits two important strategic forces present in market interactions.  

First, one can construct pools in both monopolistic and competitive settings in which the consumer lacks an incentive to separate herself from the pool. These pools are simple, do not require commitment, and depend only on willingness-to-pay rather than on intricate details of the type space. Second, when facing multiple firms, voluntary disclosure and personalized pricing can amplify competitive forces. By revealing features of one's preferences to the market, the consumer obtains a significant price concession from a less competitive firm that forces the more competitive firm to also lower its price. 

We have examined these basic strategic considerations through the lens of a stylized model, and it is worth noting important caveats to our conclusions. 

One consideration is that voluntary disclosure and personalized pricing can generate multiple equilibria. We view our analysis as identifying possibility results for when consumers can be made better off. Our results show that identifying these possibilities is subtle: simple evidence does not benefit the consumer in any equilibrium of a  monopolistic market but can do so in competitive markets. But translating these possibilities for improvement into actual gains for consumer surplus requires the coordination of their information-sharing, and here, we see an important role for intermediaries, platforms, and regulation.

A different concern is that while our model permits a consumer to disclose her valuation perfectly, consumer tracking may be limited so that at best, the consumer can disclose information about her type up to some coarse partition. While the exact analysis of such a setting would depend on the details, we would expect our main qualitative conclusions to still hold: consumer control over all-or-nothing disclosures (like simple evidence) will have a stronger effect in competitive markets than when there is a high degree of market power. 

Similarly, one may worry that, if consumers internalize a cost for ceding their privacy when they disclose information, our conclusions may not go through. In the monopolistic setting, this is true: consumers at the bottom of each messaging pool may now receive a negative payoff from remaining in their pool and prefer not to disclose at all. Ultimately, the \emph{only} strategy consistent with equilibrium is for there to be no disclosure, and so no evidence structure --- simple or rich --- could induce a consumer welfare improvement. However, this is not true of the competitive setting. Because of the price competition that disclosure can induce, the consumer always enjoys a strict improvement from disclosing information. Thus, so long as the privacy cost is sufficiently low, all of our results for competitive markets remain. Combining these observations, the addition of an exogenous value for privacy strengthens our message about the sharp differences between consumer control in monopolistic and competitive markets.

A natural question to follow asks whether the technologies envisioned by the simple and rich evidence structures in our model are feasible in real data-sharing settings. It is not difficult to imagine data-sharing tools that fit our definitions: track/do-not-track is already a part of GDPR, and a simple slider system in which consumers report the group or range that they belong to on different dimensions would fit the convexity requirement of rich evidence. Thus, practical feasibility relies on two features: the ability to transmit verifiable information, and the feasibility of implementing an equilibrium. 

As argued by \cite{goldfarbtucker} and others, an important element in the ongoing evolution of the digital economy is its increasing ability to verify information. These advances suggest that it may be technologically feasible for an intermediary --- whether a private platform or a government operated service --- to verifiably disclose aspects of consumers' preferences to sellers. Already, regulations such as GDPR have required platforms to allow consumers to opt out of being tracked. The ways in which such policies and technologies are crafted can influence not only how data is shared, but also what data will be shared in equilibrium, facilitating coordination across consumers.  Our work suggests that designing these tools with a priority to give consumers control over their data may make personalized pricing attractive and improve consumer welfare.


{\small \begin{singlespace}
    \addcontentsline{toc}{section}{References}
    \bibliographystyle{ecta}
    \bibliography{alv}
\end{singlespace} }


\appendix
\section{Appendix}

\begin{proof}[Proof of \cref{Proposition-Certification} on p. \pageref{Proposition-Certification}]
Consider an equilibrium. Let $\tilde{T}$ be the set of types that in equilibrium send the non-disclosure message, $T$. Thus, every type in $T\backslash\tilde{T}$ sends a message that fully reveals itself. Sequential rationality demands that the monopolist charges a price of $v(t)$ to every such type, leading to an interim payoff of $0$. We prove below that the non-disclosure message must induce a price that is no less than $p^*$. 

Suppose towards a contradiction that it leads to a price $\tilde{p}$ that is strictly less than $p^*$. In equilibrium, if $v(t)>\tilde{p}$, the consumer must be sending the non-disclosure message $T$ (because sending the message $\{t\}$ leads to a payoff of $0$).
Therefore, in equilibrium, 
\begin{align*}
    \tilde{T}\supseteq \{t\in T:v(t)>\tilde{p}\} \supseteq \{t\in T:v(t)\geq p^*\}.
\end{align*}
By charging a price of $\tilde{p}$, the firm's payoff is
\begin{align*}
    \tilde{p} \mu(\{t\in \tilde T: v(t)\geq \tilde{p}\})&\leq \tilde{p} \mu(\{t\in T: v(t)\geq \tilde{p}\})\\&<{p}^* \mu(\{t\in T: v(t)\geq {p}^*\})\\&={p}^* \mu(\{t\in \tilde{T}: v(t)\geq {p}^*\}),
\end{align*}
where the weak inequality follows from $\tilde{T}\subseteq T$, the strict inequality follows from $p^*$ being the (lowest) optimal price, and the equality follows from $\{t\in T:v(t)\geq p^*\}\subseteq \tilde{T}$. Therefore, the monopolist gains from profitably deviating from charging $\tilde{p}$ to a price of $p^*$ when facing the non-disclosure measure, thereby rendering a contradiction. 
\end{proof}
\begin{proof}[Proof of \cref{Proposition-MonopolistConsumerSegmentation} on p. \pageref{Proposition-MonopolistConsumerSegmentation}]
We augment the description of the strategy-profile with the off-path belief system where when the seller receives a message $M\notin\left(\cup_{s=1,\ldots,S} M_s\right)\cup M_\infty$, she puts probability $1$ on a type in $M$ with the highest valuation (i.e. a type in $\arg\max_{t\in M}v(t)$), and charges a price equal to that valuation. 

Observe that the seller has no incentive to deviate from this strategy-profile because for each (on- or off-path) message, the price that he is prescribed to charge in equilibrium is her optimal price given the beliefs that are induced by that message. 

We consider whether the consumer has a strictly profitable deviation. Let us consider on-path messages first. Consider a consumer type $t$ that is prescribed to send message $M_s$ where $p_s<v(t)\leq p_{s-1}$. Sending any message of the form $M_{s'}$ where $s'<s$ results in a higher price and therefore is not a profitable deviation. All messages of the form $M_{s'}$ where $s'>s$ are infeasible because $t\notin M_{s'}$ for any $s'>s$. Finally, if the type $t$ is such that she is prescribed to send message $M_\infty$, her equilibrium payoff is $0$, and sending any other message results in a weakly higher price. Thus, the consumer has no profitable deviation to any other on-path message. There is also no profitable deviation to any off-path message: because for any set $M$ that contains $t$, $v(t)\leq \max_{t'\in M}v(t')$, any off-path message is guaranteed to result in a payoff of $0$.
\end{proof}
\begin{proof}[Proof of \cref{Proposition-MonopolistUnidimensional} on p. \pageref{Proposition-MonopolistUnidimensional}]

Consider an equilibrium $\sigma$. Let $m^\sigma(t)$ denote the message reported by type $t$, let $F^\sigma_m\in \Delta[0,1]$ denote the firm's belief when receiving message $m$ and $\underline{t}^{\sigma}(m)$ be the lowest type in the support of that belief, and let $p^\sigma(m)$ be the sequentially rational price that he charges. In any equilibrium, $p^\sigma(m)\geq \underline{t}^{\sigma}(m)$, because otherwise the firm has a profitable deviation. We say that a message is an equilibrium-path message if there exists at least one type that sends it, and a price is an equilibrium-path price if there exists at least one equilibrium-path message that induces the firm to charge that price.
\begin{lemma}[Efficiency Lemma]\label{Lemma-MonopolistEfficiency}
	For any equilibrium $\sigma$, there exists an equilibrium that is efficient that results in the same payoff for every consumer type.  
\end{lemma}
\begin{proof}\renewcommand{\qedsymbol}{}
	Consider an equilibrium $\sigma$. Define a strategy profile $\tilde\sigma$ in which 
\begin{align*}
m^{\tilde\sigma}(t) &= \begin{cases}
                       m^\sigma(t) & \text{if }v(t)\geq p^{\sigma}(m^\sigma(t))\text{,} \\
                       \{t\} & \text{otherwise,}
                    \end{cases}\\
                   p^{\tilde\sigma}(m)&=p^\sigma(m). 
                    \end{align*}
In this disclosure strategy profile, a consumer-type that doesn't buy in equilibrium $\sigma$ is fully revealing herself in $\tilde\sigma$. Because $\sigma$ is an equilibrium, and the pricing strategy remains unchanged, such a type purchases in $\tilde\sigma$ at price $v$, and thus, efficiency is guaranteed without a change in payoffs. 

We argue that $\tilde\sigma$ is an equilibrium. Note that because $\sigma$ is an equilibrium, and we have not changed the price for any message, no consumer-type has a motive to deviate. We also argue that the monopolist has no incentive to change prices. Because $p^\sigma(m)$ is an optimal price for the firm to charge in the equilibrium $\sigma$ when receiving message $m$,
\begin{align}\label{Inequality-MonopolistPriceEfficiency}
p^\sigma(m)(1-F^\sigma_m(p^\sigma(m)))\geq p(1-F^\sigma_m(p))\text{ for every }p.
\end{align}
After receiving message $m$ in $\tilde\sigma$, the monopolist's payoff from setting a price of $p^{\tilde\sigma}(m)$ is $p^{\tilde\sigma}(m)=p^{\sigma}(m)$ (because that price is accepted for sure), and the payoff from setting a higher price is $p(1-F^{\tilde\sigma}_m(p))$. But observe that by Baye's Rule, for every $p\geq p^{\tilde\sigma}(m)$, 
\begin{align*}
	1-F^{\tilde\sigma}_m(p)=\frac{1-F^{\sigma}_m(p)}{1-F^\sigma_m(p^\sigma(m))}.
\end{align*}
Thus \eqref{Inequality-MonopolistPriceEfficiency} implies that $p^{\tilde\sigma}(m)\geq  p(1-F^{\tilde\sigma}_m(p))$ for every $p>p^{\tilde\sigma}(m)$, and clearly the monopolist has no incentive to reduce prices below $p^{\tilde\sigma}(m)$. Therefore, the monopolist has no motive to deviate.
\end{proof}
\begin{lemma}[Partitional Lemma]\label{Lemma-Partitional}
For every efficient equilibrium $\sigma$, there exists a partitional equilibrium $\tilde\sigma$ that results in the same payoff for almost every type.
\end{lemma}
\begin{proof}\renewcommand{\qedsymbol}{}
In an efficient equilibrium $\sigma$, trade occurs with probability $1$. Therefore, for every equilibrium-path message, $m$, the price charged by the monopolist after that message, $p^\sigma(m)$, must be no more than the lowest type in the support of his beliefs after receiving message $m$, $\underline{t}^\sigma(m)$ (recall that $v(t)=t$). Sequential rationality of the monopolist demands that $p^\sigma(m)$ is at least $\underline{t}^\sigma(m)$ (because charging strictly below can always be improved), and therefore, in an efficient equilibrium, $p^\sigma(m)=\underline{t}^\sigma(m)$.\smallskip

\noindent\underline{Step 1}: We first prove that the set of types being charged an equilibrium-path price $p$ is a connected set. Suppose that types $t$ and $t''>t$ are sending (possibly distinct) equilibrium-path messages $m$ and $m''$ such that $p^\sigma(m)=p^\sigma(m'')$. Because $p^\sigma(m)=\underline{t}^\sigma(m)$ and $p^\sigma(m'')=\underline{t}^\sigma(m'')$, it follows that $\underline{t}^\sigma(m)=\underline{v}^\sigma(m'')<v<v''$. Because types arbitrarily close to $\underline{t}^\sigma(m'')$ and $v''$ are both sending the message $m''$, the message $m''$ contains the interval $[\underline{t}^\sigma(m''),t'']$. 

Consider any type $t'$ in $[t,t'']$: because $[t,t'']\subseteq [\underline{t}^\sigma(m''),v'']\subseteq m''$, it follows that $m''$ is a \emph{feasible message} for type $t'$. Therefore, denoting $m'$ as the equilibrium-path message of type $t'$, type $t'$ does not have a profitable deviation to sending message $m''$ only if $p^\sigma(m')\leq p^\sigma(m)$. 

We argue that this weak inequality holds as an equality. Suppose towards a contradiction that $p^\sigma(m')< p^\sigma(m)$. Then it follows from $p^\sigma(m')=\underline{t}^\sigma(m')$ that $\underline{t}^\sigma(m')<\underline{t}^\sigma(m)\leq t\leq t'$. Therefore, the interval $[\underline{t}^\sigma(m'),t']$ is both a subset of $m'$ and contains $t$, and hence, $m'$ is a feasible message for type $t$. But then, type $t$ has an incentive to deviate from her equilibrium-path message $m$ to $m'$, which is a contradiction. \smallskip

\noindent\underline{Step 2}: For every equilibrium-path price $p$, let 
\begin{align*}
	M^\sigma(p)&\equiv\{m\in\messages:p^\sigma(m)=p\text{ and }m\text{ is an equilibrium-path message}\},\\
	T^\sigma(p)&\equiv\{t:p^\sigma(m^\sigma(t))=p\}.
\end{align*}
Observe that for every message $m$ in $M^\sigma(p)$, the monopolist's optimal price is $p$. Because the monopolist's payoff from charging any price is linear in his beliefs, and the belief induced by knowing that the type is in $T^\sigma(p)$ is a convex combination of beliefs in the set $\bigcup_{m\in M^\sigma(p)} \{F^\sigma(m)\}$, it follows that the monopolist's optimal price remains $p$ when all he knows is that the type is in $T^\sigma(p)$. 

Now consider the collection of sets
\begin{align*}
	\partition^\sigma \equiv\{m\in \messages:m=cl(T^\sigma(p))\text{ for some equilibrium-path price }p\},
\end{align*}
where $cl(\cdot)$ is the closure of a set. We argue that $\partition^\sigma$ is a partition of $[\underline{v},\overline{v}]$: clearly, $[\underline{v},\overline{v}]\subseteq \bigcup_{m\in \partition^\sigma} m$, and because each of $T^\sigma(p)$ and $T^\sigma(p')$ are connected for equilibrium-path prices $p$ and $p'$, $cl(T^\sigma(p))\bigcap cl(T^\sigma(p'))$ is at most a singleton. 

Consider a strategy-profile $\tilde\sigma$ where each type $t$ sends the message $m^{\partition^\sigma}(t)$. Fix such a message $m$ generated by $\tilde\sigma$; there exists a price $p$ that is on the equilibrium path (in the equilibrium $\sigma$) such that $m=cl(T^\sigma(p))$. Because the prior is atomless, the monopolist's optimal price when receiving message $m$ in $\tilde\sigma$ is equivalent to setting the optimal price when knowing that the type is in $T^\sigma(p)$, which as established above, is $p$. If any other message $m=[a,b]$ is reported, the monopolist believes that the consumer's type is $b$ with probability $1$. 

We argue that this is an equilibrium. We first consider deviations to other messages that are equilibrium-path for $\tilde\sigma$. For any type $t$ such that there exists a unique element in $\partition^\sigma$ that contains $t$, there exists no other feasible message that is an equilibrium-path message for $\tilde\sigma$. For any other type $t$, the strategy of sending the message $m^{\partition^\sigma}(t)$ ensures that type $t$ is sending the equilibrium-path message that induces the lower price. Finally, no type gains from sending an off-path message. Observe that all but a measure-$0$ set of types are charged the same price in $\tilde\sigma$ as they are in $\sigma$. 
\end{proof}\vspace{-0.35in}\end{proof}

\begin{proof}[Proof of \Cref{Proposition-GreedyOptimalPower} on p. \pageref{Proposition-GreedyOptimalPower}]
Since all partitional equilibria involve trade with probability $1$, a partitional equilibrium $\sigma$ has higher ex ante consumer welfare than the partitional equilibrium $\tilde\sigma$ if the average price in $\sigma$ is lower than that in $\tilde\sigma$: 
\begin{align*}
	\int_0^1 p^\sigma(m^\sigma(t))dt\leq 	\int_0^1 p^{\tilde\sigma}(m^{\tilde\sigma}(t))dt. 
\end{align*}
Thus, it suffices to prove that the greedy segmentation attains the lowest average price attainable by any partitional equilibrium.

We first describe the greedy segmentation. For a truncation of valuations $[0,v]$ where $v\leq 1$, let $p(v)$ solve $pf(p)=F(v)-F(p)$, which implies that $p({v})=\frac{{v}}{\sqrt[\leftroot{-2}\uproot{2}k]{k+1}}$; let us denote the denominator of $p(v)$ by $\gamma$, and note that $\gamma>1$. The greedy segmentation divides the $[0,1]$ interval into sets of the form $\{0\}\bigcup_{\ell=0}^\infty S_\ell$ where $S_\ell\equiv\left[\frac{1}{\gamma^{\ell+1}},\frac{1}{\gamma^\ell}\right]$. 

We prove that no partitional equilibrium generates a lower average price on the segment $S_\ell$ than $\frac{1}{\gamma^{\ell+1}}$. Consider an arbitrary $\ell\geq 0$. Consider dividing $S_\ell$ into two segments $\left[\frac{1}{\gamma^{\ell+1}},\tilde{v}\right]$ and $\left(\tilde v,\frac{1}{\gamma^{\ell}}\right]$ for some $\tilde v\in \left(\frac{1}{\gamma^{\ell+1}},\frac{1}{\gamma^{\ell}}\right)$. The higher segment is charged $\tilde v$. The lowest possible price that the lower segment is charged is $\frac{\tilde v}{\gamma}$, which is achieved if all types in $\left[\frac{\tilde v}{\gamma},\tilde v\right]$ send the same message. The resulting average price in the segment $S_\ell$ is 
\begin{align*}
	\bar{P}(\tilde{v})\equiv&(F(\tilde{v})-F(1/\gamma^{\ell+1}))\frac{\tilde v}{\gamma}+(F(1/\gamma^{\ell})-F(\tilde{v}))\tilde v\\
	=&(\tilde{v}^k-\gamma^{-k(\ell+1)})\frac{\tilde v}{\gamma}+(\gamma^{-k\ell}-\tilde{v}^k)\tilde{v}
\end{align*}
where the first equality substitutes $F(v)=v^k$. Taking derivatives,
\begin{align*}
    \frac{d^2\bar P}{d\tilde{v}^2}=(k+1)k\tilde{v}^{k-1}\left(\frac{1}{\gamma}-1\right)<0
\end{align*}
where the inequality follows from $\gamma>1$. Therefore, $\bar{P}$ is concave in $\tilde{v}$. The boundary condition that $\bar{P}(\gamma^{-\ell})=\bar{P}(\gamma^{-(\ell+1)})=\gamma^{-(\ell+1)}$ coupled with concavity of $\bar{P}$ implies that $\bar{P}(\tilde{v}) \geq \gamma^{-(\ell+1)}$ for every $\tilde{v}\in \left(\frac{1}{\gamma^{\ell+1}},\frac{1}{\gamma^{\ell}}\right)$. Therefore, no partitional equilibrium generates a lower average price than $\gamma^{-(\ell+1)}$ for the set of types in $S_\ell$. Because the greedy segmentation attains this lowerbound pointwise on every interval $S_\ell$ for every $\ell$, it is the consumer-optimal partitional equilibrium.
\end{proof}

\begin{proof}[Proof of \Cref{Proposition-DuopolistSimpleUnraveling} on p. \pageref{Proposition-DuopolistSimpleUnraveling}]
    Given a message $M$, let $\tau(i,M)\equiv\arg\min_{t\in M}|t-\ell_i|$ denote the closest type in $M$ to seller $i$; this type is well-defined because $M$ is closed. Let $\delta_t$ denote the degenerate probability distribution that places probability $1$ on type $t$. We use this notation to construct a fully revealing equilibrium:
    \begin{itemize}[noitemsep]
        \item The consumer of type $t$ always sends message $\{t\}$. 
        \item If seller $i$ receives message $M$, his beliefs are $\delta_{\tau(i,M)}$ and that the other seller has received a fully revealing message.
        \item If seller $i$ holds belief $\delta_{\tau(i,M)}$, he charges a price $p_i(M)=\max\{2\tau(i,M)\ell_i,0\}$.
        \item If $V-p_i-|t-\ell_i|>V-p_j-|t-\ell_j|$ and $V-p_i-|t-\ell_i|\geq 0$, then the consumer purchases from firm $i$. 
        \item If $V-p_L-|t-\ell_L|=V-p_R-|t-\ell_R|\geq 0$, the consumer purchases from firm $L$ if and only if $t\leq 0$, and otherwise, the consumer purchases from firm $R$.
    \end{itemize}
We argue that this is an equilibrium. Observe that each seller's on-path beliefs are consistent with Bayes rule, since $t=\tau(i,\{t\})$. In the case of an off-path message $M$, Bayes rule does not restrict the set of possible beliefs, and therefore, the above off-path belief assessment is feasible. 

To see that each firm does not wish to deviate from charging the above prices, suppose that firm $i$ receives message $M$. He believes with probability $1$ (on or off-path) that the consumer's type is $\tau(i,M)$ with probability $1$ and that the other firm $j$ has received a message $\{\tau(i,M)\}$. Denote this type by $t$. If $2t\ell_i> 0$ then the consumer is closer in location to firm $i$ and therefore $2t\ell_j < 0$. In this case, firm $i$ believes that firm $j$ is charging a price of $0$. Charging a price strictly higher than $2t\ell_i$ leads to a payoff of $0$ (because the consumer will reject such an offer and purchase instead from the other firm), and charging a price $p$ weakly below $2t\ell_i$ leads to a payoff of $p$ (because the consumer always breaks ties in favor of the closer firm). Therefore, firm $i$ has no incentive to deviate. If $2t\ell_i\leq 0$, then the consumer is located closer to the other firm $j$ and is being charged a price equal to $2t\ell_j$. In this case, charging any strictly positive price leads to a payoff of $0$ (because the consumer will purchase the good from the other firm). Therefore, in either case, firm $i$ has no incentive to deviate. 

Finally, we argue that the consumer has no incentive to deviate. By sending a fully revealing message, $\{t\}$, the consumer obtains an equilibrium payoff of $V-(|t|+1)$. If $t\leq 0$, the consumer obtains a price of $0$ from firm $R$, which is the lowest possible price. Therefore, there is no incentive to send any other message to firm $R$. Sending any other message $M\in M(t)$ to firm $L$ induces a weakly higher price because for any feasible message $M\in M(t)$, $\tau(L,M)\leq t$, and therefore, $2\tau(L,M)\ell_L\geq 2t\ell_L$. Thus, the consumer has no strictly profitable deviation from sending any other message $M\in M(t)$ to firm $L$. An analogous argument implies that if the consumer's type is $t>0$, she also does not gain from deviating to any other feasible disclosure strategy. 

\end{proof}

\begin{proof}[Proof of \Cref{Proposition-DuopolistSimpleExistence} on p. \pageref{Proposition-DuopolistSimpleExistence}]
We first show given the pricing strategies that the consumer has no incentive to deviate. 

Consider a consumer type $t$ such that $t\in (t_1^L,t_1^R)$, or in other words, $t\ell_i<t_1^i\ell_i$. The equilibrium strategies are that the consumer sends the message $\{t\}$ to each firm. Given these equilibrium strategies, the consumer is quoted a price of $\max\{2t\ell_i,0\}$ by firm $i$. If the consumer deviates and sends message $[-1,1]$ to firm $i$, she induces a price of $p_1^i=2t_1^i\ell_i$, which is strictly higher. Therefore, this deviation is not strictly profitable. 

Now suppose that $t\ell_i\geq t_1^i\ell_i $. The equilibrium strategies are that the consumer sends message $[-1,1]$ to firm $i$ and $\{t\}$ to firm $j$. Because the consumer, in equilibrium, is quoted a price of $0$ by firm $j$, sending the other message cannot lower that price. Given the equilibrium message, the consumer is quoted a price of $2t_1^i\ell_i$ by firm $i$, and deviating leads to a weakly higher price of $2t\ell_i$. Therefore, this deviation is not strictly profitable. 

We now consider whether firm $i$ has an incentive to deviate. It follows from the proof of \Cref{Proposition-DuopolistSimpleUnraveling} that the prices are optimal whenever firm $i$ receives an (equilibrium-path) message of $\{t\}$ for $t\in (t_1^L,t_1^R)$. An identical argument applies when firm $i$ receives an (off-path) message of $\{t\}$ for $t\ell_i\geq t_1^i\ell_i $: in this case, firm $i$ believes that firm $j$ is charging a price of $0$, and thus, the optimal price is $2t\ell_i$ (because the consumer always breaks ties in favor of firm $i$). When firm $i$ receives an (equilibrium-path) message of $\{t\}$ for $t\ell_j\geq t_1^j\ell_j $, firm $i$ believes that firm $j$ is charging a price of $2t_1^j$. The equilibrium prescribes that firm $i$ charges a price of $0$, which leads to a payoff of $0$ (because the consumer breaks ties in favor of firm $j$), and any strictly positive price also leads to a payoff of $0$. Finally, consider the case when firm $i$ receives an (equilibrium-path) message of $[-1,1]$. Firm $i$ infers that $t\ell_i\geq t_1^i\ell_i$ and believes that firm $j$ is charging a price of $0$. Because $p_1^i$ is, by definition, a profit-maximizing price in response to a price of $0$, firm $i$ has no strictly profitable deviation.
\end{proof}

\begin{proof}[Proof of \Cref{Proposition-DuopolistRichExistence} on p. \pageref{Proposition-DuopolistRichExistence}]
We use an off-path belief system where if firm $i$ receives an off-path message $M$, she holds degenerate beliefs $\delta_{\tau(i,M)}$ that put probability $1$ on type $\tau(i,M)$ where recall that $\tau(i,M)$ is defined as the type in $M$ that is located closest to firm $i$ (this was defined in the proof of \Cref{Proposition-DuopolistSimpleUnraveling}). Given such beliefs, the firm charges a price $p_i(M)=\max\{2\tau(i,M)\ell_i,0\}$ for an off-path message $M$. 

First, we prove that given the pricing strategies, no consumer has an incentive to deviate. Consider a consumer type $t$ such that $t^i_s\ell_i<t\ell_i\leq t^i_{s-1}\ell_i$ for some $s=1,2,\ldots$. Such a consumer should be sending message $M^i_s$ to both firms. Such a message induces a price of $0$ from firm $j$ and $p_s^i=2t_s^i\ell_i$ from firm $i$. No message can induce a lower price from firm $j$. Therefore, any strictly profitable deviation must induce a strictly lower price from firm $i$. We show that this is not possible. 

We first argue that the consumer does not have a profitable deviation to any other equilibrium-path message. Suppose that $t\ell_i< t^i_{s-1}\ell_i$. In this case, $M^i_s$ is the only equilibrium-path message that type $t$ can send to firm $i$. If $t\ell_i= t^i_{s-1}\ell_i$, then type $t$ can send either message $M^i_s$ or $M^i_{s-1}$ but because $p_s^i\leq p_{s-1}^i$, this is not a strictly profitable deviation. 

We now argue that the consumer does not have a profitable deviation to any off-path message. Any feasible message $M\in \messages(t)$ satisfies the property that the closest type in $M$ to firm $i$ is at least as close as $t$ to firm $i$; or formally: $t\ell_i\leq \tau(i,M)\ell_i$. In that case, the price that the consumer is charged is $2\tau(i,M)\ell_i\geq 2t\ell_i>2t^i_s\ell_i=p_s^i$. Therefore, this deviation is not strictly profitable.  

Finally, we argue that the firms have no incentive to deviate in their pricing strategies. For any equilibrium-path message, the prices charged by firms are (by construction) equilibrium prices. For any off-path message $M$, each firm assumes that the consumer sent the equilibrium-path message to the other firm. If $\tau(i,M)\ell_i>0$ then firm $i$ assumes that firm $j$ is charging a price of $0$, and then charging a price of $2\tau(i,M)\ell_i$ is a best-response (assuming that the consumer breaks ties in favor of the closer firm). If $\tau(i,M)\ell_i\leq 0$, then firm $i$ believes that the consumer is being charged a price $p^j_s$ by firm $j$ for some $s$ where $t^i_s\ell_i<\tau(i,M)\ell_i\leq t^i_{s-1}\ell_i$. Because the consumer breaks ties in favor of the closer firm, firm $i$ anticipates that the consumer will reject any strictly positive price.  
\end{proof}

\begin{proof}[Proof of \Cref{Proposition-DuopolyWelfare} on p. \pageref{Proposition-DuopolyWelfare}]
We show that $p^*>p^i_1$ for every $i$. Observe that
\begin{align*}
    p^L_1=\frac{2F(-p^L_1/2)}{f(-p^L_1/2)}<\frac{2F(0)}{f(0)}=p^*
\end{align*}
where the first equality follows from the first-order condition that $p^L_1$ solves, the inequality follows from $F$ being strictly log-concave, and the second equality follows from the definition of $p^*$. A symmetric argument shows that $p^R_1<p^*$. 

Now we prove that all consumers are better off in the equilibrium we construct in the game with simple evidence (\Cref{Proposition-DuopolistSimpleExistence}). All types where $t\ell_i\geq t_1^i\ell_i$ are buying the good at a lower price because $p^i_1<p^*$. Consider any other type, i.e., where $t\ell_i <t_1^i\ell_i$ for every $i\in \{L,R\}$. Suppose that $t\ell_i > 0$. That type in equilibrium buys the good from firm $i$ at the price $2t\ell_i<2t_1^i\ell_i$, which equals $p^i_1$. Therefore, it obtains the good at a lower price than $p^*$. Finally, if $t=0$, that type obtains the good at a price equal to $0$. 

An analogous argument ranks prices relative to the equilibrium constructed in the game with rich evidence (\Cref{Proposition-DuopolistRichExistence}). All types in that equilibrium pay a price that is less than $p_1^i$, and therefore, buy the good at a price lower than $p^*$. 
\end{proof}

\begin{proof}[Proof of \Cref{Proposition-OligopolyRichExistence} on p. \pageref{Proposition-OligopolyRichExistence}]
We begin by proving that each $M^{ij}_s$ is convex for each $s$. If $t$ and $t'$ are each in $M^{ij}_s$, then $t_i \geq t_j \geq \max_{k\in N\backslash\{i,j\}}t_k$, and $t_i' \geq t_j' \geq \max_{k\in N\backslash\{i,j\}}t_k'$. Therefore, for every $\rho\in (0,1)$, the following are true:
\begin{align*}
    \rho t_i+(1-\rho)t_i'\geq \rho t_j+(1-\rho)t_j'&\geq \rho \max_{k\in N\backslash\{i,j\}}t_k+(1-\rho)\max_{k\in N\backslash\{i,j\}}t_k'\geq \max_{k\in N\backslash\{i,j\}} (\rho t_k+(1-\rho) t_k')\\
    \rho t_i+(1-\rho)t_i'-\rho t_j+(1-\rho)t_j'&= \rho(t_i-t_j)+(1-\rho)(t_i-t_j) \leq \rho p^{ij}_{s-1}+(1-\rho)p^{ij}_{s-1}=p^{ij}_{s-1}.
\end{align*}
Therefore, $\rho t+(1-\rho)t'$ is also an element of $M^{ij}_{s}$, and so $M^{ij}_s$ is convex for every $s$. Clearly, $M^{ij}_\infty$ is also convex by the same argument. 

The remainder of this argument follows the same logic as that of \Cref{Proposition-DuopolistRichExistence}. If firm $i$ receives an off-path message $M$, she holds degenerate beliefs $\delta_{\tau(i,M)}$ where type $\tau(i,M)$ is the type in $M$ that has the highest value of $t_i-\max_{j\neq i} t_j$. Because $M$ is a closed set, $\tau(i,M)$ is defined for every $M$. Given such beliefs, the firm charges a price $p_i(M)=\max\{t_i-\max_{j\neq i} t_j,0\}$. 

First, we prove that given the pricing strategies, no consumer type has an incentive to deviate. Consider a consumer type such that $t\in M^{ij}_s$ and $t_i-t_j>p^{ij}_s$. A message of $M^{ij}_s$ to each firm induces all firms other than firm $i$ to set a price of $0$ and induces firm $i$ to set a price of $p^{ij}_s$. No other message can lead to lower prices from any other firm. Moreover, the consumer cannot send any other equilibrium path message to firm $i$ that leads to lower prices. Finally, every off-path message sent to firm $i$ can only increase the price since for any off-path message $M\in \messages(t)$, $p_i(M)\geq M^{ij}_s$. 

Second, we argue that firms have no incentives to deviate in their pricing strategies. For any on-path message, the prices charged are (by construction) optimal. For any off-path message $M$, each firm assumes that the consumer sent the equilibrium path message to other firms. If Firm $i=\alpha(\tau(i,M))$ (i.e., Firm $i$ is the favorite), then it assumes that other firms are charging a price of $0$, in which case $p_i(M)$ is optimal. If Firm $i$ is not the favorite, then it charges a price of $0$, and anticipates any strictly positive price to be rejected. 
\end{proof}

\begin{proof}[Proof of \Cref{Proposition-OligopolyWelfare} on p. \pageref{Proposition-OligopolyWelfare}]
Because $f$ is log-concave, it follows directly from \cite{caplin1991aggregation} that for each firm $i$, $Q^i(p)$ is log-concave. Since firms are symmetric, we drop the subscripts. Let $p^*$ denote the price from the symmetric equilibrium of the game without personalized pricing. Observe that for each firm $i$,
\begin{align*}
    p^i_1 = -\frac{Q(p^i_1)}{q(p^i_1)}<-\frac{Q(0)}{q(0)}=p^*,
\end{align*}
where the first equality follows from the first-order condition that $p^i_1$ solves, the inequality follows from $Q$ being strictly log-concave, and the second equality follows from the first-order condition that $p^*$ solves. Now all consumers are necessarily better off in the equilibrium constructed with simple evidence. The argument is the same as in \Cref{Proposition-DuopolyWelfare}: all consumer types either pay a price of $p^i_1$ or below. In the equilibrium constructed with rich evidence, it follows from symmetry that $p^{ij}_1 = p^{i}_1$, and therefore, once again, all consumer types either pay a price of $p^i_1$ or below. 
\end{proof}

\end{document}